\documentclass[%
 reprint,
showpacs,preprintnumbers,
 amsmath,amssymb,
 aps,
]{revtex4-1}

\usepackage{graphicx}
\usepackage{dcolumn}
\usepackage{bm}

\usepackage{natbib}

\usepackage{tikz}
\usetikzlibrary{arrows,automata}

\usepackage{dsfont}
\usepackage{algorithmicx}
\usepackage{algpseudocode}
\usepackage{caption}
\usepackage{subcaption}
\usepackage{amsthm}

\newtheorem{theorem}{Theorem}[section]
\newtheorem{lemma}{Lemma}[section]

\newtheorem{definition}[theorem]{Definition}

\begin{document}

\preprint{APS/123-QED}

\title{Cascading Edge Failures: A Dynamic Network Process}
\thanks{This work is partially supported by NSF grants CCF1011903 and CCF1513936.}%

\author{June Zhang}
 
\author{Jos\'{e} M.F.~Moura}%
\affiliation{%
 Department of Electrical and Computer Engineering\\
 Carnegie Mellon University\\
 Pittsburgh, PA, USA\\
}%

\date{\today}

\begin{abstract}
This paper considers the dynamics of edges in a network. The Dynamic Bond Percolation (DBP) process models, through stochastic local rules, the dependence of an edge $(a,b)$ in a network on the states of its neighboring edges. Unlike previous models, DBP does not assume statistical independence between different edges. In applications, this means for example that failures of transmission lines in a power grid are not statistically independent, or alternatively, relationships between individuals (dyads) can lead to changes in other dyads in a social network. We consider the time evolution of the probability distribution of the network state, the collective states of all the edges (bonds), and show that it converges to a stationary distribution. We use this distribution to study the emergence of global behaviors like consensus (i.e., catastrophic failure or full recovery of the entire grid) or coexistence (i.e., some failed and some operating substructures in the grid). In particular, we show that, depending on the local dynamical rule, different network substructures, such as hub or triangle subgraphs, are more prone to failure. 
\begin{description}
\item[PACS numbers]
\end{description}
\end{abstract}

\maketitle

\section{Introduction}
In 2003, the failure of a single power line triggered a series of cascading failures throughout the Northeastern United States and Southeastern Canada \cite{minkel20082003}. Connectivity of the power grid led to an uncontrollable propagation of failures that resulted in over 50 millions customers losing power. Similarly, the rapid spread of infection during epidemics is due to contacts between individuals of a population. These phenomena are rare and difficult to study via experiments. Due to their scale, they are also expensive to simulate. It is therefore advantageous to be able to study such phenomena analytically. 

Cascading failures and epidemics can be modeled as \emph{network processes} \textemdash dynamical processes whose substrate is a particular network structure. It is of interest to understand the dynamics of such processes, particularly the impact of the network structure \cite{Jackson, newman2010networks}. In some applications such as epidemiology, where nodes represent individuals, it is more intuitive to study network processes on nodes \cite{eames2002modeling, JZhangJournal, JZhangJournal2}. This means that the state of each node changes according to deterministic or stochastic dynamics often depending on the states of the nodes' neighbors. This inclusion couples the evolution of each node with that of all the other nodes in the network, resulting in a complex dynamical system. For modeling phenomena such as cascading failures of transmission lines in the power grid or the evolution of relationships, called dyads, in social networks, it is more appropriate to study edge-centric network processes where the edges change state according to some dynamics.

Dynamic edge models have been used in computational sociology to study the time evolution of dyads \cite{wasserman1980analyzing, leenders1995models, snijders2005models, pepyne2007topology}. However, some models assume that the dynamics of different dyads are statistically independent. These models can not capture effects such as triadic closure, where friendships are more likely to form between individuals with common neighbors \cite{granovetter1973strength}. Alternatively, models try to capture interdependencies by incorporating knowledge of network structure into the dynamics of a dyad. This assumes that individual dyads have full knowledge of the overall network structure, which may not be feasible in practice. 

The bond percolation model from statistical mechanics has been used to study the robustness and resilience of network structures to stochastic bond (i.e., edge) removals \cite{broadbent1957percolation, newman2010networks, callaway2000network}. However, the standard percolation model is \emph{not} a network process, as it does not model dynamical evolution of bond states over time \cite{steif2009survey}. Additionally, the bond percolation model assumes that the size of the network is infinite, and account for the topology only through the degree distribution and not higher order degree correlations. Lastly, percolation is a \emph{macroscopic} model; it characterizes global properties such as the percentage of removed components, but can not provide information on microscopic properties such as probability that a set of bonds is removed. 

The Dynamic Bond Percolation (DBP) model we present in this paper differs from previous models by the inclusion of local coupling dependencies in time between edges as represented by a heterogenous network structure. Each edge in a network can be in one of two states: open and closed. The dynamics of edges are no longer independent as an edge $(a,b)$ changes state according to a stochastic rule depending only on the states of its local neighboring edges. For example, this means that the failure rate of a transmission line is affected by the failures of other transmission lines. Alternatively, the formation of friendship links between individuals $a$ and $b$ depends on the other relationships of $a$ and $b$. 


The state of DBP at time $t$ is the collection of the states of all the edges at time $t$. We show that the probability distribution over the states of DBP reaches an equilibrium distribution. For certain local rules, we can compute this distribution explicitly for any finite-size, unweighted, undirected network structure, unlike previous models that approximate the underlying network by simpler structures (e.g., complete graphs) or infinite-size networks (i.e., mean-field approximation) \cite{liggett1999stochastic, barrat2008dynamical}. Analysis of the equilibrium distribution informs us of the emergence of global behavior. For example, in certain parameter regimes, the most probable network state at equilibrium is that of consensus, denoting either complete failure or complete recovery of all transmission lines. In other parameter regimes, it provides information about the subset of transmission lines more prone to cascading failures or more likely to form relationships. We will see that these effects depend on the local dynamical rules. 

When the imbalance between node $a$ and node $b$ does not matter, the critical edges tend to belong to star subgraphs (i.e., hub structures). Whereas critical edges tend to belong to triangle subgraphs and other network structures (e.g., $P_3$ or $P_4$, see Section~\ref{sec:graphreview}) when imbalance or mutual relationships matter. We illustrate our analysis within two real-world networks: 118-node IEEE test bus power grid \cite{power:Online} and 198-node social network  \cite{weeks2002social}.

Section~\ref{sec:model} explains the Dynamic Bond Percolation process in detail. Section~\ref{sec:equilibriumdistribution} derives the closed-form equilibrium distribution for 3 different dynamical rules. Section~\ref{sec:mostprobable} describes the Most-Probable Network Problem. Section~\ref{sec:regimeI} through Section~\ref{sec:regimeIIIedge} analyzes the solution space of the Most-Probable Network Problem for the different parameter regimes. Section~\ref{sec:conclusion} concludes the paper.


\section{Dynamic Bond Percolation Process}\label{sec:model}


Consider a population of $N$ individuals or components represented as nodes in a network. The adjacency matrix, $\mathbf{A}_{\max}$, describes the largest set of potentially open or closed bonds between the nodes. We will refer to the network represented by $\mathbf{A}_{\max}$ the \emph{maximal network}. For example, the maximal network may be a representation of the underlying power grid or a social network between $N$ individuals. We assume that the maximal network is a simple, connected, undirected graph. Let $E_{\max}$ denote the set of edges and $V_{\max}$ denote the set of nodes in the maximal network respectively. Figure~\ref{fig:rnetwork} shows an example of a maximal network. 

Edges in the maximal network may change state over time. For instance, transmission lines in the power grid can fail, social ties between individuals can weaken or strengthen with time. We are particularly interested in scenarios where there is a contagion aspect to the dynamical process: a single transmission line failure can lead to other line failures such as in blackouts. 

\begin{figure}
\center
\includegraphics[width=2in]{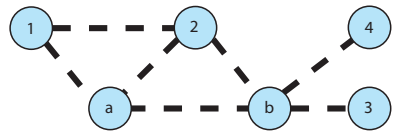}
\caption{Maximal network represented by $\mathbf{A}_{\max}$. Dashed edges (i.e., bonds) are the only possible edges in the network.}
\label{fig:rnetwork}
\end{figure}

In this paper, we present the Dynamic Bond Percolation (DBP) process, a dynamical extension of the bond percolation model \cite{broadbent1957percolation}. Edges in the maximal network \emph{open} or \emph{close} according to stochastic dynamical rules specified by DBP. We assume that the close of an edge depends on the state of the neighboring edges, thereby coupling the underlying network topology with the process dynamics. In applications, DBP can be used to model cascading failures of transmission lines in the power grid or the formation and dissolution of ties in a social network. Figure~\ref{fig:evolutionx(t)} shows a possible realization of the DBP process on the maximal network shown in Figure~\ref{fig:rnetwork}.

The state of the DBP process at time $t \ge 0 $, $A(t)$, is represented by the $N \times N$ adjacency matrix $ \mathbf{A}$, where
\begin{align*}
\mathbf{A}_{i,j} = & 1 \text{ if edge $(a,b)$ is close}\\
                       = & 0  \text{ if edge $(a,b)$ is open}.
\end{align*}

We will call $\mathbf{A}$ the \emph{network state}. The set of nodes in the network state $V(\mathbf{A}) = V_{\max}$. The set of edges in the network state corresponds to the set of closed edges in the maximal network; therefore, $E(\mathbf{A}) \subseteq E_{\max}$. The space of all possible network states is $\mathcal{A}$. Since each edge in the maximal network can be either open or closed, then $|\mathcal{A}| = 2^{|E_{\max}|}$. DBP makes the following assumptions:

\textbf{Assumption 1:} Network states that contain edges not in $E_{\max}$ are invalid. For example, with respect to the maximal network in Figure~\ref{fig:rnetwork}, the network state in Figure~\ref{fig:invalidnetwork} is not valid since edge $(2,4)$ is not in the maximal network.

\begin{figure}
\center
\includegraphics[width=2in]{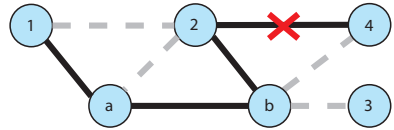}
\caption{Invalid network state. Edge (2,4) is not in the maximal network and can not change state. Solid edges are \emph{closed}. Dashed edges are \emph{open}.}
\label{fig:invalidnetwork}
\end{figure}

\begin{figure*}
	\centering
        \begin{subfigure}[b]{0.22\textwidth}
                \centering
                \includegraphics[width =\textwidth]{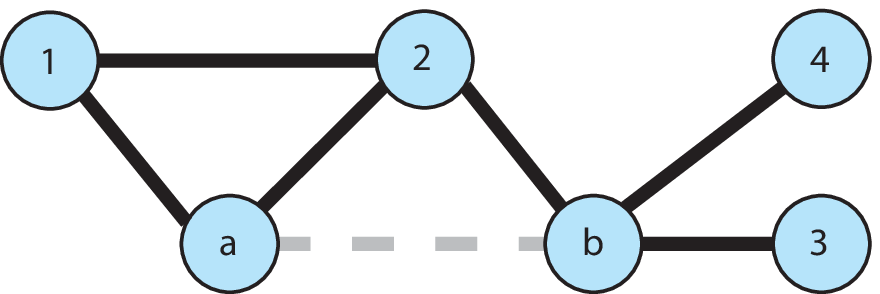}
                \caption{$X(t_1)$}
                \label{fig:x1}
        \end{subfigure}%
        ~ 
        \begin{subfigure}[b]{0.22\textwidth}
                \centering
                \includegraphics[width =\textwidth]{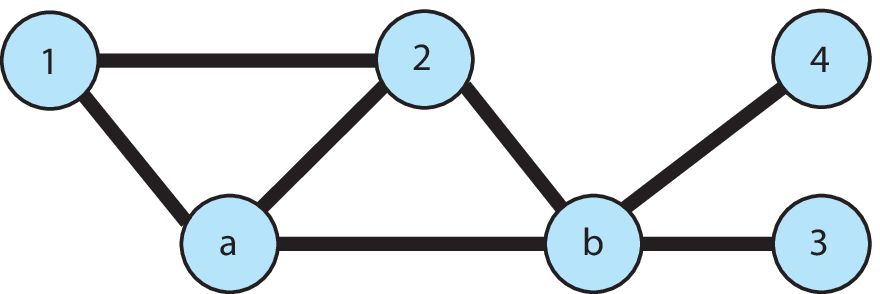}
                \caption{$X(t_2)$}
                \label{fig:x2}
        \end{subfigure}
        ~ 
           \begin{subfigure}[b]{0.22\textwidth}
                \centering
                \includegraphics[width =\textwidth]{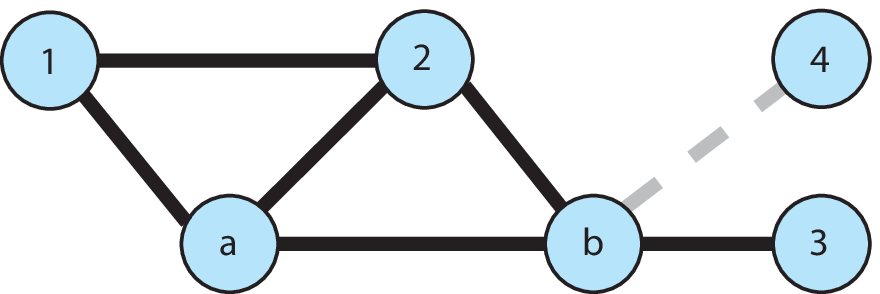}
                \caption{$X(t_3)$}
                \label{fig:x3}
        \end{subfigure}
        ~ 
            \begin{subfigure}[b]{0.22\textwidth}
                \centering
                \includegraphics[width =\textwidth]{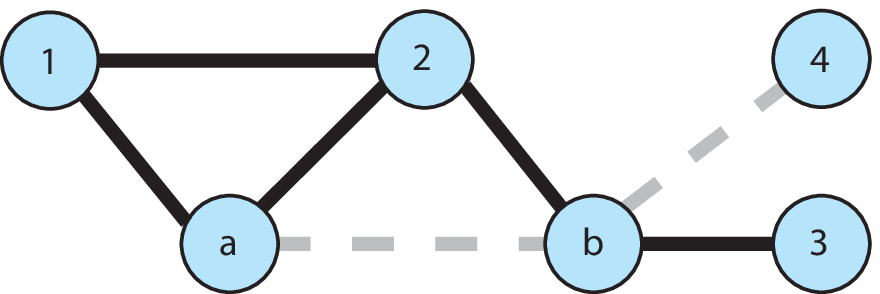}
                \caption{$X(t_4)$}
                \label{fig:x4}
        \end{subfigure}
        ~ 
        \caption{Example evolution of $A(t)$. Solid edges are \emph{closed}. Dashed edges are \emph{open}.}\label{fig:evolutionx(t)}
\end{figure*}


\textbf{Assumption 2:} Multiple edge opening or closure can not occur simultaneously. Only one edge changes state per transition.

\textbf{Assumption 3:} The time it takes for an edge $(a,b)$ to go from \emph{close} to \emph{open} (e.g., $t_4 - t_3$ as in Figure~\ref{fig:evolutionx(t)}) is exponentially distributed with rate
\begin{equation}\label{eq:dissipationrate}
q(\mathbf{A}, T_{a,b}^-\mathbf{A}) = \mu'.
\end{equation}

This means that the probability that edge $(a,b)$ switches from \emph{close} to \emph{open} $\Delta t$ time units in the future  is
\[
P(A(t + \Delta t) = T_{a,b}^-\mathbf{A} \mid A(t) = \mathbf{A}) \approx \mu' \Delta t + o(\Delta t),
\]
where $o(\cdot)$ is the little-o notation. We call $\mu' > 0$ the edge opening rate. DBP assumes that $\mu'$ is the same for all closed edges. For applications where edge opening is rare, $\mu'$ can be arbitrarily small as long as it is not 0.

\textbf{Assumption 4:} The time it takes for an edge $(a,b)$ to go from \emph{open} to \emph{close} (e.g., $t_2 - t_1$ as in Figure~\ref{fig:evolutionx(t)}) is exponentially distributed with rate
\begin{equation}\label{eq:formationrate}
q(\mathbf{A}, T_{a,b}^+\mathbf{A}) = \lambda'\gamma'^{f(\mathcal{N}_a,\mathcal{N}_b)},
\end{equation}
where $\mathcal{N}_a$ and $\mathcal{N}_b$ denotes the set of closed edges connected to node $a$ and the set of closed edges connected to node $b$ in network state $\mathbf{A}$, respectively. We call the function $f(\mathcal{N}_a,\mathcal{N}_b)$ the \emph{cascade function}. It captures how the edge closure rate of $(a,b)$ depends on the local neighborhood of edge $(a,b)$. In this paper, we consider the following cascade functions:

\begin{enumerate}

\item SUD-DBP (Sum-Dependent Dynamic Bond Percolation):
\begin{equation}\label{eq:SUDform}
 f(\mathcal{N}_a,\mathcal{N}_b) = |\mathcal{N}_a| + |\mathcal{N}_b|
\end{equation}
\item TRI-DBP (Triangle-Dependent Dynamic Bond Percolation):
\begin{equation}\label{eq:TRIform}
f(\mathcal{N}_a,\mathcal{N}_b) = |\mathcal{N}_a \cap \mathcal{N}_b|.
\end{equation}
\item POD-DBP (Product-Dependent Dynamic Bond Percolation): 
\begin{equation}\label{eq:PODform}
f(\mathcal{N}_a,\mathcal{N}_b) = |\mathcal{N}_a||\mathcal{N}_b|.
\end{equation}
\end{enumerate}

SUD-DBP assumes that the closure rate of an edge depends on the sum of the number of closed neighboring edges. POD-DBP assumes that the closure rate depends on the product of the number of closed neighboring edges. These model different dynamics because POD-DBP implicitly accounts for imbalance between the number of closed edges on node $a$, $|\mathcal{N}_a|$, and node $b$, $|\mathcal{N}_b|$. For example, consider the two scenarios in Figure~\ref{fig:scenarios}. Under SUD-DBP, the closure rate of edge $(a,b)$ is $\lambda'\gamma'^{6}$ for both Scenarios A and B. Under POD-DBP, the closure rate of edge $(a,b)$ is $\lambda'\gamma'^{9}$ for Scenario A and $\lambda'\gamma'^{5}$ for Scenario B.

TRI-DBP assumes that the closure rate of an edge depends on the number of closed neighboring edges sharing a common agent. This comes from the concept of triadic closure, which states that, for 3 agents $a, b$, and $c$, if there is a connection between $(a,b)$ and $(a,c)$, then it is more likely that there is a connection between $(b,c)$ \cite{brandes2013studying}. Under TRI-DBP, the closure rate of edge $(a,b)$ is $\lambda'$ for both Scenario A and $\lambda'\gamma'$ for Scenario B. 

When $f(\mathcal{N}_a,\mathcal{N}_b)  = 0$, the transition rate $q(\mathbf{A}, T_{a,b}^+\mathbf{A}) = \lambda'$. This means that it it is possible for an edge $(a,b)$ to close independent of the state of neighboring edges. Therefore, we consider $\lambda'$ to be the spontaneous edge closure rate and $\gamma'$ to be the cascading edge closure rate.

\begin{figure*}
	\centering
        \begin{subfigure}[b]{0.5\textwidth}
                \centering
                \includegraphics[width=2in]{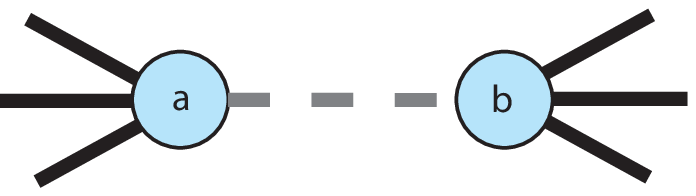}
                \caption{Scenario A: $|\mathcal{N}_a|=3, |\mathcal{N}_b| = 3$.}
                \label{fig:3to3}
        \end{subfigure}%
        ~ 
        \begin{subfigure}[b]{0.5\textwidth}
                \centering
                \includegraphics[width=1.6in]{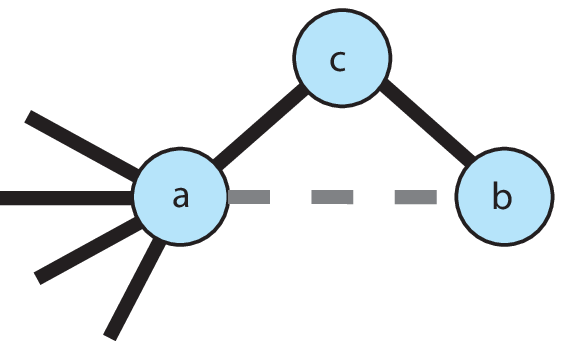}
                \caption{Scenario B: $|\mathcal{N}_a| = 5, |\mathcal{N}_b| = 1$.}
                \label{fig:5to1}
        \end{subfigure}

        \caption{Different edge removal scenarios. Solid edges are \emph{closed}. Dashed edges are \emph{open}.}
        \label{fig:scenarios}
\end{figure*}

Since the transition rates are not time dependent, DBP is a time-homogenous process. Since the model assumes both spontaneous edge opening and close, there is no absorbing state. Under the stated assumptions, the Dynamic Bond Percolation model, $\{A(t), t \ge 0\}$, is an irreducible, continuous-time Markov process with finite state space, $\mathcal{A} = \{\mathbf{A}\}$; each state in the Markov process corresponds to a potential network state $\mathbf{A}$. The dimension of the configuration space is $|\mathcal{A}| = 2^{|E_{\max}|}$. 

\section{Time-Asymptotic Behavior}\label{sec:equilibriumdistribution}


While it is difficult to completely characterize the dynamics of the Dynamic Bond Percolation process, its time-asymptotic behavior $\left(\text{i.e., } \lim_{ t \to \infty} A(t) \right)$ can be studied using its equilibrium distribution, $\pi$. The equilibrium distribution of DBP is a probability mass function (PMF) over $\mathcal{A}$. Since DBP is a finite-state, continuous-time Markov process, the equilibrium distribution $\pi$ is unique \cite{Kelly}. It can be found by solving the left eigenvalue-eigenvector problem 
\[
\pi \mathbf{Q} = 0,
\] 
where $\mathbf{Q} $ is the transition rate matrix, also known as the infinitesimal matrix \cite{norris1998markov}.

The matrix $\mathbf{Q}$ characterizes the transition rates between all the states in $\mathcal{A}$ using \eqref{eq:dissipationrate}, \eqref{eq:formationrate}. For DBP, it is an asymmetric $2^{|E_{\max}|} \times 2^{|E_{\max}|}$ matrix. Element $\mathbf{Q}_{ij}$ corresponds to the transition rate between 2 states $\mathbf{i}, \mathbf{j} \in \mathcal{A}$ where $i$ and $j$ are decimal scalar representations of the network states $\mathbf{i}$ and $\mathbf{j}$, respectively.
 
The equilibrium distribution of a continuous-time Markov process is the left eigenvector of the transition rate matrix corresponding to the zero eigenvalue. However, the challenge of finding $\pi$ is that the dimensions of $\mathbf{Q}$ scales exponentially with the total number of edges in the maximal network, $|E_{\max}|$. This makes computing the equilibrium distribution prohibitively expensive for large networks. In the next section, we will show that we can avoid this computation for SUD-DBP, TRI-DBP, and POD-DBP by finding the equilibrium distribution, $\pi$, up to a constant, in closed form.

\subsection{Review of Graph Theoretic Concepts}\label{sec:graphreview}
First, we review graph theory terms necessary for the rest of the paper \cite{west2001introduction, Jackson}.

\begin{definition}
A walk is a list $v_0, e_1, v_1, e_2, \ldots, e_k, v_k$ of vertices and edges. The length of the walk is the number of edges in the list. The number of walks in an undirected graph from node $i$ to node $j$ of length $k$ is 
\begin{equation*}\label{eq:numwalk}
(\mathbf{A}^k)_{i,j},
\end{equation*}
where $\mathbf{A}^k$ is the adjacency matrix of an undirected, unweighted graph raised to the $k$th power.

\end{definition}
\begin{definition}
A path is a walk where all the vertices are distinct (although some literature does not make this distinction between paths and walks). A graph that is a path is called a path graph and written as $P_n$, where $n$ is the number of vertices (not edges) in the path. By convention, the path graph $P_n$ is equivalent to a path of length $n-1$. Figure~\ref{fig:P3} shows the $P_3$ subgraph and Figure~\ref{fig:P4} shows the $P_4$ subgraph.
\end{definition}

\begin{definition}
A cycle is a path where the endpoints $v_0 = v_k$. A graph that is a cycle is called a cycle graph and written as $C_n$, where $n$ is the number of vertices (not edges) in the cycle. By convention, the cycle graph $C_n$ is equivalent to a cycle of length $n$. Figure~\ref{fig:C3} shows the $C_3$ subgraph. $C_3$ subgraphs are also called triangles.
\end{definition}

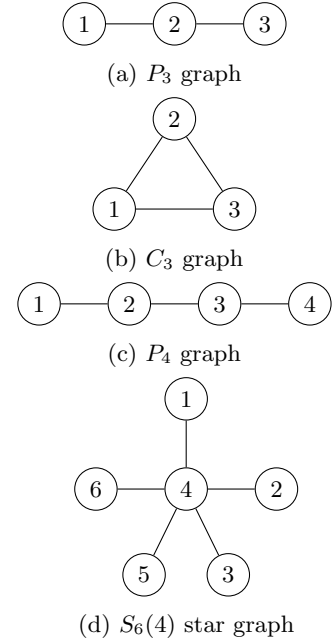
\begin{figure}[h]
        \centering
             \begin{subfigure}[b]{0.2\textwidth}
                \centering
       \begin{tikzpicture}
  [scale=.8,auto=left,every node/.style={circle,draw= black}]

  \node (n1) at (1,1) {1};
  \node (n2) at (2.5,1)  {2};
  \node (n3) at (4,1)  {3};
 
  \foreach \from/\to in {n1/n2,n2/n3}
    \draw (\from) -- (\to);

\end{tikzpicture}
                \caption{$P_3$ graph}
                \label{fig:P3}
        \end{subfigure}

 \begin{subfigure}[b]{0.3\textwidth}
                \centering
              \begin{tikzpicture}
  [scale=.8,auto=left,every node/.style={circle,draw= black}]

  \node (n1) at (1,1) {1};
  \node (n2) at (2 ,2.5)  {2};
  \node (n3) at (3,1)  {3};
 
  \foreach \from/\to in {n1/n2,n2/n3, n1/n3}
    \draw (\from) -- (\to);

\end{tikzpicture}
                \caption{$C_3$ graph}
                \label{fig:C3}
        \end{subfigure}

        \begin{subfigure}[b]{0.3\textwidth}
                \centering
           \begin{tikzpicture}
  [scale=.8,auto=left,every node/.style={circle,draw= black}]

  \node (n1) at (1,1) {1};
  \node (n2) at (2.5,1)  {2};
  \node (n3) at (4,1)  {3};
  \node (n4) at (5.5,1)  {4};
 
  \foreach \from/\to in {n1/n2,n2/n3,n3/n4}
    \draw (\from) -- (\to);

\end{tikzpicture}
                \caption{$P_4$ graph}
                \label{fig:P4}
                 \end{subfigure}%

        \begin{subfigure}[b]{0.5\textwidth}
                \centering
           \begin{tikzpicture}
  [scale=.8,auto=left,every node/.style={circle,draw= black}]

  \node (n1) at (0,1.5) {1};
  \node (n2) at (1.5,0)  {2};
  \node (n3) at (0.7,-1.4266)  {3};
  \node (n4) at (0,0)  {4};
\node (n5) at (-0.7,-1.4266)  {5};
   \node (n6) at (-1.5,0)  {6};
 
  \foreach \from/\to in {n1/n4,n2/n4,n3/n4, n5/n4, n6/n4}
    \draw (\from) -- (\to);

\end{tikzpicture}
                \caption{$S_6(4)$ star graph}
                \label{fig:stargraph}

        \end{subfigure}%
 
        \caption{$P_3$, $C_3$, $P_4$ and $S_6(4)$ subgraphs}\label{fig:p4c3}
\end{figure}

\begin{definition}
A star graph, $S_n(i)$, has $n$ vertices that are only connected to the center vertex $i$. Figure~\ref{fig:stargraph} shows the $S_6(4)$ subgraph.
\end{definition}

\begin{definition}\label{def:matching}\cite{west2001introduction}
A matching, $\mathcal{M}$, of the graph $G(V,E)$, also called the Independent Edge Set, is a subset of edges $E$ such that no vertex in $V$ is incident to more than one edge in $\mathcal{M}$; see Figure~\ref{fig:matching}. Maximum Matching is a matching with the maximum number of edges; see Figure~\ref{fig:maxmatching}. 
\end{definition}

The number of edges in the maximum matching is known as the matching number, $\nu(G)$. 

In this paper, we introduce a generalization to the matching set.
\begin{definition}\label{def:starmatching}
A star matching, $\mathcal{S}$, of the graph $G(V,E)$, is a subset of edges $E$ such that these edges form a collection of disconnected star graphs; see Figure~\ref{fig:starmatching}. Maximum star matching is a star matching with the maximum number of edges; see Figure~\ref{fig:maxstarmatching}. Note that $\mathcal{M} \subset \mathcal{S}$.
\end{definition}

\begin{figure*}
	\centering
        \begin{subfigure}[b]{0.37\textwidth}
                \centering
                \includegraphics[width =\textwidth]{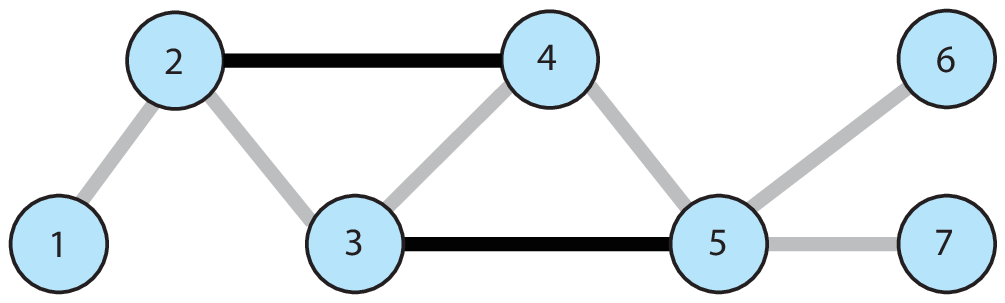}
                \caption{Matching: $\{(2,4), (3,5)\}$}
                \label{fig:matching}
        \end{subfigure}%
        ~ 
        \begin{subfigure}[b]{0.37\textwidth}
                \centering
                \includegraphics[width =\textwidth]{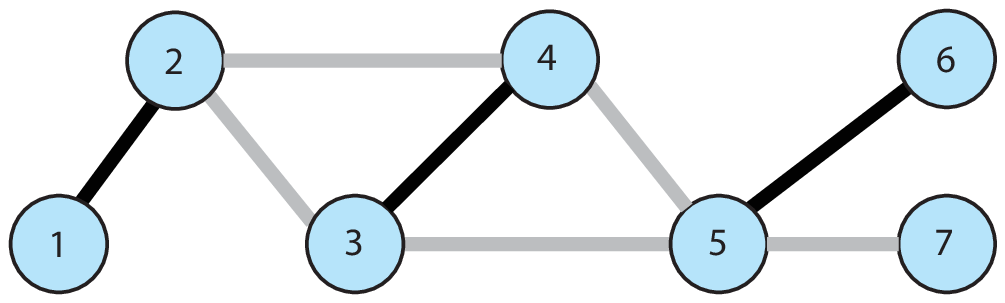}
                \caption{Max Matching: $\{(1,2), (3,4), (5,6)\}$}
                \label{fig:maxmatching}
        \end{subfigure}
        ~ 
          
           \begin{subfigure}[b]{0.37\textwidth}
                \centering
                \includegraphics[width =\textwidth]{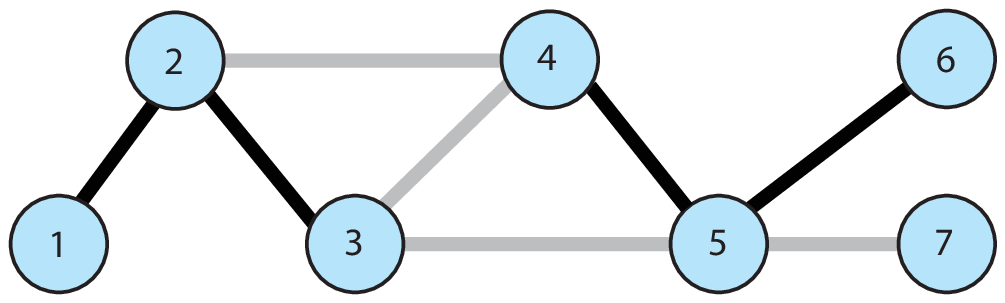}
                \caption{Star Matching: $\{(1,2),(2,3), (4,5), (5,6)\}$}
                \label{fig:starmatching}
        \end{subfigure}
        ~ 
            \begin{subfigure}[b]{0.37\textwidth}
                \centering
                \includegraphics[width =\textwidth]{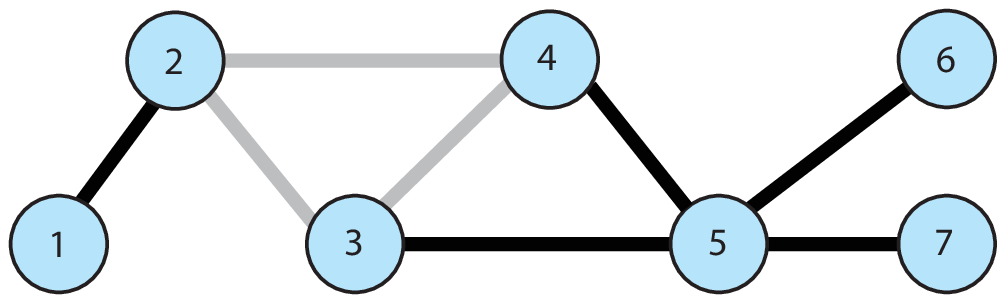}
                \caption{Max Star Matching: $\{(1,2), (3,5),(4,5), (5,6), (6,7) \}$}
                \label{fig:maxstarmatching}
        \end{subfigure}
        ~ 
        \caption{Matching and Star Matching}\label{fig:matchingall}
\end{figure*}

\subsection{Reversibility and Equilibrium Distribution}\label{sec:equilibrium}
Some Markov processes possess the property that the process forward in time is statistically equivalent to the process backward in time. These Markov processes are called \emph{reversible} Markov processes. The following theorem is important in deriving the stationary distribution of a reversible Markov process:

\begin{theorem}[From \cite{Kelly}]\label{reversible}
A stationary Markov process is reversible if and only if there exists a collection of positive number $\pi(j), j \in \mathcal{L}$, summing to unity that satisfy the detailed balance conditions
\[
\pi(j)q(j,k) = \pi(k)q(k,j), \quad j,k, \in  \mathcal{L},
\]
where $q(\cdot,\cdot)$ is the transition rate of the Markov process. When there exists such a collection $\pi(j), j \in \mathcal{L}$, it is the equilibrium distribution of the process.
\end{theorem}

Using this theorem, we will show that the equilibrium distributions for SUD-DBP, POD-DBP, and TRI-DBP have the form:
\begin{equation}
\pi(\mathbf{A}) = \frac{1}{Z} \left(\frac{\lambda'}{\mu'}\right)'^{|E(\mathbf{A})|} \gamma'^{g(E(\mathbf{A}))}, \quad \mathbf{A} \in \mathcal{A},
\end{equation}
where the partition function, $Z$, is 
\begin{equation}\label{eq:partition}
Z = \sum_{\mathbf{A} \in \mathcal{A}}\left(\frac{\lambda'}{\mu'}\right)^{|E(\mathbf{A})|}\gamma'^{g(E(\mathbf{A}))}.
\end{equation}
The exponent $|E(\mathbf{A})|$ is the total number of closed edges in a network state, and $g(E(\mathbf{A}))$ is the number of special network structures induced by the set of closed edges $E(\mathbf{A})$. These special network structures are derived for SUD-DBP, POD-DBP, TRI-DBP in sections~\ref{sec:SUB-DBP}, ~\ref{sec:POD-DBP}, and~\ref{sec:TRI-DBP},  respectively.

The equilibrium distribution is the product of three terms: the partition function $Z$, a \emph{structure-free} term, and a \emph{structure-dependent} term that depends on the maximal network. Let $E(\mathbf{A})$ denote the set of closed edges in network state $\mathbf{A}$. The term $\left(\frac{\lambda'}{\mu'}\right)^{|E(\mathbf{A})|}$ is structure-free because it depends only on the number of close edges in network state whereas the term $\gamma'^{g(E(\mathbf{A}))}$ depends on the underlying maximal network topology.


\subsection{(SUD-DBP) Sum-Dependent Dynamic Bond Percolation Model}\label{sec:SUB-DBP}

\begin{theorem}\label{thm:SUD}
The Sum-Dependent Dynamic Bond Percolation model, $\{A(t), t \ge 0\}$, is a reversible Markov process and the equilibrium distribution is
\begin{equation}\label{eq:sumeq}
\pi(\mathbf{A}) = \frac{1}{Z}\left(\frac{\lambda'}{\mu'}\right)^{|E(\mathbf{A})|}\gamma'^{g(E(\mathbf{A}))}, \quad \mathbf{A} \in \mathcal{A},
\end{equation}
where $g(E(\mathbf{A}))$ it the number of $P_3$ subgraphs induced by the set of closed edges $E(\mathbf{A})$.

The number of \emph{close} edges in a network state is 
\begin{equation*}
|E(\mathbf{A})| = \frac{\mathbf{1}^T\mathbf{A}\mathbf{1}}{2},
\end{equation*}
where $\mathbf{1} = [1,1,\ldots,1]^T$.

The number of $P_3$ subgraphs induced by the edges $E(\mathbf{A})$ in a network state is
\begin{equation}\label{eq:gSUD}
g(E(\mathbf{A})) =  \sum_{i=1}^N \sum_{j > i} (\mathbf{A}^2)_{i,j} = \sum_{i=1}^N \binom{k_i}{2},
\end{equation}
where $k_i = \sum_{j = 1}^N \mathbf{A}_{ij}$ and is the number of closed edges at node $i$. The derivation of the closed-from equation of $g(E(\mathbf{A}))$ is in Appendix~\ref{sec:gsud}

\end{theorem}

In the SUD-DBP model, the sufficient statistics are the total number of close edges and the total number of paths of length 2 (i.e., the number of $P_3$ subgraphs) induced by the closed edges. The number of $P_3$ subgraphs induced by the closed edges is related also to the degree of the network state $\mathbf{A}$. Surprisingly, this means that a sufficient statistic of a process where edges change state in time is a nodal characteristic. 

\begin{proof}
We prove that the equilibrium distribution $\pi(\mathbf{A})$ of the SUD-DBP is given by \eqref{eq:sumeq}. 

By Theorem~\ref{reversible}, $\pi(\mathbf{A})$ satisfies the detailed balance equations:
\begin{equation}\label{eq:detailedbalance12}
\pi(\mathbf{A})q(\mathbf{A}, T^+_{a,b}\mathbf{A}) = \pi(T^+_{a,b}\mathbf{A})q(T^+_{a,b}\mathbf{A}, \mathbf{A})
\end{equation}
and
\begin{equation}\label{eq:detailedbalance22}
\pi(\mathbf{A})q(\mathbf{A}, T^-_{a,b}\mathbf{A}) = \pi(T^-_{a,b}\mathbf{A})q(T^-_{a,b}\mathbf{A}, \mathbf{A}).
\end{equation}

We prove~\eqref{eq:detailedbalance12} first. 

Let $T^+_{a,b}\mathbf{A}$ denote the network state that is the same as network state $\mathbf{A}$ except edge $(a,b)$ switches from open to close. According to the transition rate in \eqref{eq:formationrate} for SUD-DBP \eqref{eq:SUDform}, 
\[
q(\mathbf{A}, T^+_{a,b}\mathbf{A}) = \lambda'\gamma'^{k_a+k_b},
\]
where $k_a = |\mathcal{N}_a| =  \sum_{i = 1}^N \mathbf{A}_{ai}$ and $k_b = |\mathcal{N}_b| = \sum_{i = 1}^N \mathbf{A}_{bi}$ are the number of closed edges incident at node $a$ and $b$, respectively.

According to the equilibrium distribution of SUD-DBP \eqref{eq:sumeq}, the equilibrium probability of network state $\mathbf{A}$ is
\begin{equation*}
\pi(\mathbf{A}) = \frac{1}{Z}\left(\frac{\lambda'}{\mu'}\right)^{|E(\mathbf{A})|}\gamma'^{g(E(\mathbf{A}))},
\end{equation*}
where $g(E(\mathbf{A}))$ is the number of $P_3$ subgraphs induced by the set of closed edges $E(\mathbf{A})$. The LHS of \eqref{eq:detailedbalance12} is

\begin{equation}
\begin{split}
&\pi(\mathbf{A})q(\mathbf{A}, T^+_{a,b}\mathbf{A}) \\
&=  \frac{1}{Z}\left(\frac{\lambda'}{\mu'}\right)^{|E(\mathbf{A})|}\gamma'^{g(E(\mathbf{A}))}\left(\lambda' \gamma'^{k_ak_b} \right)\\
& =  \frac{1}{Z}\left(\frac{\lambda'^{|E(\mathbf{A})| + 1}}{\mu'^{|E(\mathbf{A})|}}\right)\gamma'^{g(E(\mathbf{A})) + k_a+k_b}.
\end{split}
\end{equation}

According to the transition rate in \eqref{eq:dissipationrate},
\[
q(T^+_{a,b}\mathbf{A}, \mathbf{A}) = \mu'.
\]
Recognize that by definition, $|E(T^+_{a,b}\mathbf{A})| = 1 + |E(\mathbf{A})|$. Consider the closure of edge $(a,b)$. This means that the number of paths of length 2 from any node in $\mathcal{N}_a$ to the node $b$ is $k_a$ and the number of paths of length 2 from the node $a$ to any node in $\mathcal{N}_b$ in $k_b$. Therefore, $g(E(T^+_{a,b}\mathbf{A})) = g(E(\mathbf{A}))+ k_a + k_b$.

The RHS of \eqref{eq:detailedbalance12} is

\begin{equation}
\begin{split}
&\pi(T^+_{a,b}\mathbf{A})q(T^+_{a,b}\mathbf{A}, \mathbf{A}) \\
&= \frac{1}{Z}\left(\frac{\lambda'}{\mu'} \right)^{|E(T^+_{a,b}\mathbf{A})|}\gamma'^{g(E(T^+_{a,b}\mathbf{A}))}\mu'\\
&= \frac{1}{Z}\left(\frac{\lambda'^{|E(\mathbf{A})| + 1}}{\mu'^{|E(\mathbf{A})|}}\right)\gamma'^{g(E(\mathbf{A})) + k_a+k_b}.
\end{split}
\end{equation}

The LHS and RHS of \eqref{eq:detailedbalance12} are equivalent. Similar reasoning holds for \eqref{eq:detailedbalance22}. Since the detailed balance equations are satisfied by \eqref{eq:sumeq}, Theorem~\ref{reversible} proves Theorem~\ref{thm:SUD}.

\end{proof}

\subsection{(TRI-DBP) Triangle-Dependent Dynamic Bond Percolation Model}\label{sec:TRI-DBP}

\begin{theorem}\label{thm:TRI}
The Triangle-Dependent Dynamic Bond Percolation model, $\{A(t), t \ge 0\}$, is a reversible Markov process and the equilibrium distribution is
\begin{equation}\label{eq:TRIeq}
\pi(\mathbf{A}) = \frac{1}{Z}\left(\frac{\lambda'}{\mu'}\right)^{|E(\mathbf{A})|}\gamma'^{g(E(\mathbf{A}))}, \quad \mathbf{A} \in \mathcal{A},
\end{equation}
where $g(E(\mathbf{A}))$ it the number of $C_3$ subgraphs induced by the set of closed edges $E(\mathbf{A})$.

The number of \emph{close} edges in a network state is 
\begin{equation*}
|E(\mathbf{A})| = \frac{\mathbf{1}^T\mathbf{A}\mathbf{1}}{2},
\end{equation*}
where $\mathbf{1} = [1,1,\ldots,1]^T$.

The number of $C_3$ subgraphs induced by the edges $E(\mathbf{A})$ in a network state is \cite{west2001introduction}
\begin{equation}\label{eq:TRI}
g(E(\mathbf{A})) =  \sum_{i = 1}^N \frac{(\mathbf{A}^3)_{i,i}}{6}.
\end{equation}

\end{theorem}

In the TRI-DBP model, the sufficient statistics are the total number of close edges and the total number of triangles (i.e., the number of $C_3$ subgraphs) induced by the closed edges. The proof for Theorem~\ref{thm:TRI} follows the same steps as the proof for Theorem~\ref{thm:SUD} except 1) the transition rate $q(\mathbf{A}, T^+_{a,b}\mathbf{A})$ for TRI-DBP is given by \eqref{eq:TRIform}, and 2) the number of $C_3$ subgraphs  $g(E(T^+_{a,b}\mathbf{A})) = |\mathcal{N}_a \cap \mathcal{N}_b| + g(E(\mathbf{A}))$.

\subsection{(POD-DBP) Product-Dependent Dynamic Bond Percolation Model}\label{sec:POD-DBP}

\begin{theorem}\label{thm:POD}
The Product-Dependent Dynamic Bond Percolation model, $\{A(t), t \ge 0\}$, is a reversible Markov process and the equilibrium distribution is
\begin{equation}\label{eq:podeq}
\pi(\mathbf{A}) = \frac{1}{Z}\left(\frac{\lambda'}{\mu'}\right)^{|E(\mathbf{A})|}\gamma'^{g(E(\mathbf{A}))}, \quad \mathbf{A} \in \mathcal{A},
\end{equation}
where $\mathbf{A}$ is the adjacency matrix, $g(E(\mathbf{A}))$ is the number of triangles $C_3$, and the paths of length 3, $P_4$ formed by the set of closed edges, $E(\mathbf{A})$.

The number of \emph{close} edges, $|E(\mathbf{A})|$, is
\begin{equation*}
|E(\mathbf{A})| = \frac{\mathbf{1}^T\mathbf{A}\mathbf{1}}{2},
\end{equation*}
where $\mathbf{1} = [1,1,\ldots,1]^T$.

The number of  $C_3$ and $P_4$ subgraphs is
\begin{equation}\label{eq:POD}
\begin{split}
&g(E(\mathbf{A})) = \sum_{i = 1}^N \frac{(\mathbf{A}^3)_{i,i}}{6} + \\
& \sum_{i = 1}^N \sum_{j > i} \left[ (\mathbf{A}^3)_{i,j} - (\mathbf{A}_{i,j})\left(  (\mathbf{A}^2)_{i,i} + \sum_{k = 1, k \neq i,j}^N \mathbf{A}_{k,j}\right) \right].
\end{split}
\end{equation}
The derivation of the closed-from equation of $g(E(\mathbf{A}))$ is in Appendix~\ref{sec:gpod}.

\end{theorem}

In the POD model, the sufficient statistics are the total number of closed edges and the total number of $C_3$ and $P_4$ subgraphs induced by the closed edges. The POD model and the SUD model do not have the same sufficient statistics. The proof for Theorem~\ref{thm:POD} follows the same steps as the proof for Theorem~\ref{thm:SUD} except 1) the transition rate $q(\mathbf{A}, T^+_{a,b}\mathbf{A})$ for POD-DBP is given by \eqref{eq:PODform}, and 2) the number of paths of length 3 from any node in $\mathcal{N}_a$ to any node in $\mathcal{N}_b$ through edge $(a,b)$ is $|\mathcal{N}_a||\mathcal{N}_b|$. Therefore $g(E(T^+_{a,b}\mathbf{A})) = k_ak_b + g(E(\mathbf{A}))$.

\section{Critical Structures and the Most-Probable Network Problem}\label{sec:mostprobable}

In the second part of the paper, we will use DBP to study critical structures in the maximal networks. Assuming that DBP models cascading failures of transmission lines in the power grid, we are interested in understanding the interaction between the underlying power grid structure and the closure (i.e., failure) and opening (i.e., recovery) rates. Alternatively, DBP may model formation of social ties between individuals; the critical structures are then relationships that are more likely to be formed and maintained.

The most-probable network state is the network state in $\mathcal{A}$ with the highest equilibrium probability:

\begin{equation}\label{eq:astar}
\mathbf{A}^* = \arg \max_{\mathbf{A} \in \mathcal{A}} \pi(\mathbf{A}) = \arg \max_{\mathbf{A} \in \mathcal{A}}  \left(\frac{\lambda'}{\mu'}\right)^{|E(\mathbf{A})|}\gamma'^{g(E(\mathbf{A}))}.
\end{equation}

Once the DBP process has reached equilibrium, it is the network state most likely to be observed. We call \eqref{eq:astar} the Most-Probable Network Problem and $\mathbf{A}^*$ the \emph{most-probable network}. By evaluating $\mathbf{A}^*$, we see that some edges are more prone to closure than others depending on the dynamic parameters, $\lambda', \mu', \gamma'$ and the maximal network topology.


Further, the Most-Probable Network Problem does not depend on finding the partition function, $Z$, \eqref{eq:partition}. This means we do not have to sum over $2^{|E_{\max}|}$ configurations. We will prove in later section that the most-probable network configuration for SUD-DBP, TRI-DBP, and POD-DBP can be found using polynomial-time algorithms. We can partition the parameter space of DBP into four regimes and determine in each regime the most-probable network.
\begin{description}
\item[Regime I)]Recovery Dominant: $0< \frac{\lambda'}{\mu'} \leq 1, 0 < \gamma'\leq 1$
\item[Regime II)]Cascading Failure:  $0 < \frac{\lambda'}{\mu'} \leq 1, \gamma'> 1$
\item[Regime III)] Cascading Prevention: $\frac{\lambda'}{\mu'} > 1, 0 < \gamma'\leq 1$
\item[Regime IV)] Removal Dominant: $\frac{\lambda'}{\mu'} > 1,  \gamma'>1$.
\end{description}
We are especially interested in contrasting the most-probable networks of SUD-DBP, TRI-DBP, and POD-DBP models in the different parameter regimes. This will give us insight in how the cascade function $f(\mathcal{N}_a, \mathcal{N}_b)$ affects the dynamic process.

\section{Regime I) Recovery Dominant and Regime IV) Removal Dominant}\label{sec:regimeI}

In Regime I) Recovery Dominant, the structure-free and the structure dependent terms are both decreasing exponential functions of the number of closed edges. Therefore, the most-probable network for SUD-DBP, TRI-DBP, and POD-DBP is the network with no edges, $\mathbf{A}_{0} = \{\mathbf{A} \in \mathcal{A} : |E(\mathbf{A})| = 0 \}$; for this regime, the opening rate is high enough that the most-probable network is a consensus state such that none of the edges in the maximal network are considered at-risk.

In Regime IV), the topology dependent and the topology independent terms are both increasing exponential functions.  Therefore, the most-probable network for SUD-DBP, TRI-DBP, and POD-DBP is $\mathbf{A}_{\max} = \{\mathbf{A} \in \mathcal{A} : |E(\mathbf{A})| = E_{\max} \}$; for this regime, the closure rate is high enough that the most-probable network is a consensus state such that all of the edges in the maximal network are considered at-risk.

The most-probable network for SUD-DBP, TRI-DBP, and POD-DBP are the same in Regime I) and Regime IV). We will show in the next section that it may be different for these models in Regime II) and Regime III) as in these regimes, there is competition between the structure-free term and the structure-dependent term .

\section{Regime II) Cascading Failure}
In Regime II) Cascading Failure: $0 < \frac{\lambda'}{\mu'} \leq 1, \gamma'> 1$, the structure-free term is driven by edge opening and the structure-dependent term is driven by edge closure. For SUD-DBP, TRI-DBP, and POD-DBP, we expect the solution space of the Most-Probable Network Problem \eqref{eq:astar} to exhibit phase transition depending on if edge opening or edge closure dominates. From the analysis of Regime I) and Regime IV), we expect that when the closure process dominates, the most-probable network of SUD-DBP, TRI-DBP, and POD-DBP will be driven toward $\mathbf{A}_{\max}$; whereas if the opening process dominates, the most-probable network for both models will be driven toward $\mathbf{A}_{0}$. 

Unlike Regime I) and IV), there may be solutions to the Most-Probable Network Problem that are neither $\mathbf{A}_{0}$ nor $\mathbf{A}_{\max}$. We call these solutions the \emph{non-degenerate} most-probable networks. The existence of these solutions means that subset of edges in the maximal network are more at-risk of closure (i.e., failure) than other edges during cascading failures.

%

To find these non-degenerate solutions, we have to solve the Most-Probable Network Problem \eqref{eq:astar}, a combinatorial optimization problem. In general, such computation is NP-hard \cite{boros2002pseudo}. For SUD-BDP, POD-BDP, and TRI-BDP however, the Most-Probable Network Problem can be solved exactly using polynomial-time algorithm using submodularity \cite{grotschel1981ellipsoid}.


Recall the definition of a submodular function:

\begin{definition}[\cite{billionnet1985maximizing}]\label{def:submodular}
A set function, $h: \mathcal{P}(E) \to \mathcal{R}$, is submodular if and only if for any $E(\mathbf{A}_1), E(\mathbf{A}_2) \subseteq E$ with $E(\mathbf{A}_2) \subseteq E(\mathbf{A}_1), i \not \in E \setminus E(\mathbf{A}_1)$:
\begin{equation}\label{eq:satsub}
h(E(\mathbf{A}_1) \cup \{i\}) - h(E(\mathbf{A}_1)) \leq h(E(\mathbf{A}_2) \cup \{i\}) - h(E(\mathbf{A}_2)).
\end{equation}
\end{definition}

Submodular functions are closed under nonnegative linear combination \cite{schrijver2003combinatorial}. Consider the function
\[
h(E) = \sum_{i=1}^M \alpha_i f_i(E).
\]
If $\alpha_i \ge 0 \, \forall i =1,\ldots, M$ and the functions $f_i(E)$ are submodular, then $h(E)$ is also submodular.


First, we need the following lemma:

\begin{lemma}\label{lemma:networksubmodular}
Consider two sets of closed edges, $E(\mathbf{A}_1), E(\mathbf{A}_2) \subseteq E_{\max}$ and an additional closed edge $ i \in E_{\max} \setminus \{E(\mathbf{A}_1) \cup E(\mathbf{A}_2)\}$. For a given maximal network, the number of subgraphs induced by the edges in $E(\mathbf{A}_1)$ and $E(\mathbf{A}_2)$ are $g(E(\mathbf{A}_1)) \ge 0$ and $g(E(\mathbf{A}_1) \ge 0$, respectively. Let the number of subgraphs induced by the edges $E(\mathbf{A}_1) \cup \{i\} = g(E(\mathbf{A}_1)) + m_1$ and the edges $E(\mathbf{A}_2) \cup \{i\} = g(E(\mathbf{A}_2)) + m_2$; therefore $m_1 \ge 0 $ is the number of additional induced subgraphs created with the inclusion of edge $i$ in $E(\mathbf{A}_1)$ and $m_2 \ge 0$ is the number of additional induced subgraphs created with the inclusion of edge $i$ in $E(\mathbf{A}_2)$. If $E(\mathbf{A}_2) \subseteq E(\mathbf{A}_1)$, then:

\begin{enumerate}
\item $g(E(\mathbf{A}_1)) \ge g(E(\mathbf{A}_2)) \ge 0$,
\item $m_1 \ge  m_2 \ge 0$.
\end{enumerate}
\end{lemma}

\begin{proof}\label{proof:lemmasubmodular}
\begin{enumerate}
\item When $E(\mathbf{A}_2) \subset E(\mathbf{A}_1)$, edges in $E(\mathbf{A}_2)$ can not induce more subgraphs  than edges in $E(\mathbf{A}_1)$. When $E(\mathbf{A}_2) = E(\mathbf{A}_1)$, then the edges in $E(\mathbf{A}_1)$ and $E(\mathbf{A}_2)$ will induce the same number of subgraphs. Hence, $g(E(\mathbf{A}_1)) \geq g(E(\mathbf{A}_2)) \ge 0$.

\item Every edge in $E(\mathbf{A}_2)$ is also an edge in $E(\mathbf{A}_1)$. Therefore, adding edge $i$ to $E(\mathbf{A}_2)$ will generate the same or less number of subgraphs as adding edge $i$ to $E(\mathbf{A}_1)$. Hence, $m_1 \ge m_2 \ge 0$.

\end{enumerate}
\end{proof}


Using Lemma~\ref{lemma:networksubmodular}, we can prove that

\begin{theorem}\label{theorem:networksubmodular}
In Regime II), for SUD-DBP, TRI-DBP, and POD-DBP, the function  
\[
-\log(Z\pi(\mathbf{A})) = -|E(\mathbf{A})| \log \left(\frac{\lambda'}{\mu'} \right) -  g(E(\mathbf{A})) \log(\gamma').
\]
is submodular. The most-probable network, 
\[
\mathbf{A}^* = \arg \max \pi(\mathbf{A}) = \arg \min -\log(Z\pi(\mathbf{A})), 
\]
is the minimum of a submodular function and can therefore be computed in polynomial-time \cite{grotschel1981ellipsoid} .

\end{theorem}

\begin{proof}\label{proof:networksubmodular}
By the additive closure property of submodular functions, the function $-\log(Z\pi(\mathbf{A}))$ is submodular when
\[
f_1(E(\mathbf{A})) = -|E(\mathbf{A})| \log \left(\frac{\lambda'}{\mu'} \right)
\]
and
\[
f_2(E(\mathbf{A})) =  - g(E(\mathbf{A})) \log(\gamma')
\]
are submodular functions in regime II).

We need to show that $f_1(\mathbf{A})$ and $f_2(\mathbf{A})$ satisfy Definition~\ref{def:submodular} when $0< \frac{\lambda'}{\mu'} \le 1$ and $\gamma' > 1$.

Consider $E(\mathbf{A}_1), E(\mathbf{A}_2) \subseteq E$ with $E(\mathbf{A}_2) \subseteq E(\mathbf{A}_1), i \not \in E \setminus E(\mathbf{A}_1)$.

\begin{equation}\label{eq:proofsub1}
\begin{split}
&f_1(E(\mathbf{A}_1) \cup \{i\}) - f_1(E(\mathbf{A}_1) ) \\
&= -(|E(\mathbf{A}_1)| + 1)\log \left(\frac{\lambda'}{\mu'} \right) + |E(\mathbf{A}_1)|\log \left(\frac{\lambda'}{\mu'} \right) \\
&= -\log \left(\frac{\lambda'}{\mu'} \right).
\end{split}
\end{equation}

\begin{equation}\label{eq:proofsub2}
\begin{split}
&f_1(E(\mathbf{A}_2) \cup \{i\}) - f_1(E(\mathbf{A}_2) )\\ 
&= -(|E(\mathbf{A}_2)| + 1)\log \left(\frac{\lambda'}{\mu'} \right) + |E(\mathbf{A}_2)|\log \left(\frac{\lambda'}{\mu'} \right) \\
&= -\log \left(\frac{\lambda'}{\mu'} \right).
\end{split}
\end{equation}

Equations \eqref{eq:proofsub1} and~\eqref{eq:proofsub2} satisfy the submodular condition \eqref{eq:satsub} when $0< \frac{\lambda'}{\mu'} \le 1$. In fact, function $f_1(E(\mathbf{A}))$ is always submodular regardless of $\frac{\lambda'}{\mu'}$.

\begin{equation}\label{eq:proofsub3}
\begin{split}
&f_2(E(\mathbf{A}_1) \cup \{i\}) - f_2(E(\mathbf{A}_1) )\\ 
&= -(g(E(\mathbf{A}_1)) + m_1) \log(\gamma') + g(E(\mathbf{A}_1) \log(\gamma')\\ 
& = - m_1 \log(\gamma').
\end{split}
\end{equation}

\begin{equation}\label{eq:proofsub4}
\begin{split}
&f_2(E(\mathbf{A}_2) \cup \{i\}) - f_2(E(\mathbf{A}_2) )\\ 
&= -(g(E(\mathbf{A}_2)) + m_2) \log(\gamma') + g(E(\mathbf{A}_2) \log(\gamma')\\ 
& = - m_2 \log(\gamma').
\end{split}
\end{equation}

According to Lemma~\ref{lemma:networksubmodular}, $m_1 \ge m_2 \ge 0$. The function $f_2(E(\mathbf{A})$ satisfies the definition of submodularity when $\gamma' > 1$.

Furthermore, since Lemma~\ref{lemma:networksubmodular} is true for all subgraphs induced by $E(\mathbf{A})$. It is applicable for SUD-DBP, TR-DBP, and POD-DBP.

\end{proof}

\subsection{Power Grid Example}

\begin{figure*}
\begin{subfigure}{0.3\linewidth}
\centering
\includegraphics[width=0.95\linewidth]{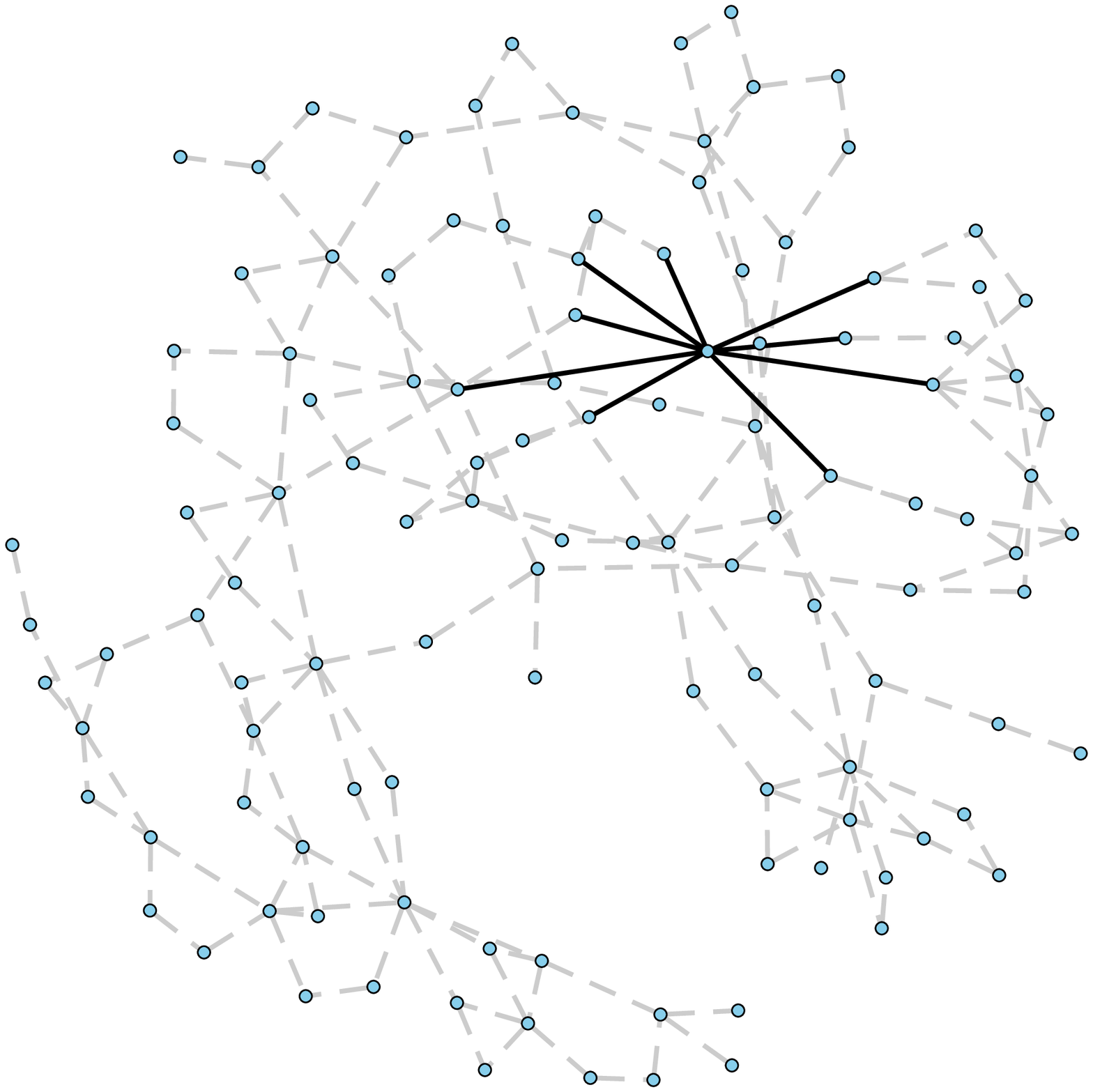}
 \captionsetup{format = hang, justification = centering}
\caption{SUD-DBP: $|E(\mathbf{A}^*)| = 9$ \\$\frac{\lambda'}{\mu'} = 0.0028, \gamma'= 4.4222$}
\label{fig:118_SUDresultMatrix_lau0.0027826_mu1_g4.4222_E9}
\end{subfigure}\hfill
\begin{subfigure}{0.3\linewidth}
\centering
\includegraphics[width=0.95\linewidth]{7a.eps}
 \captionsetup{format = hang, justification = centering}
 \caption{SUD-DBP: $|E(\mathbf{A}^*)| = 0$ \\$\frac{\lambda'}{\mu'} = 0.0167, \gamma'= 1.844$}
 \label{fig:118_A0}
\end{subfigure}\hfill
\begin{subfigure}{0.3\linewidth}
\centering
\includegraphics[width=0.95\linewidth]{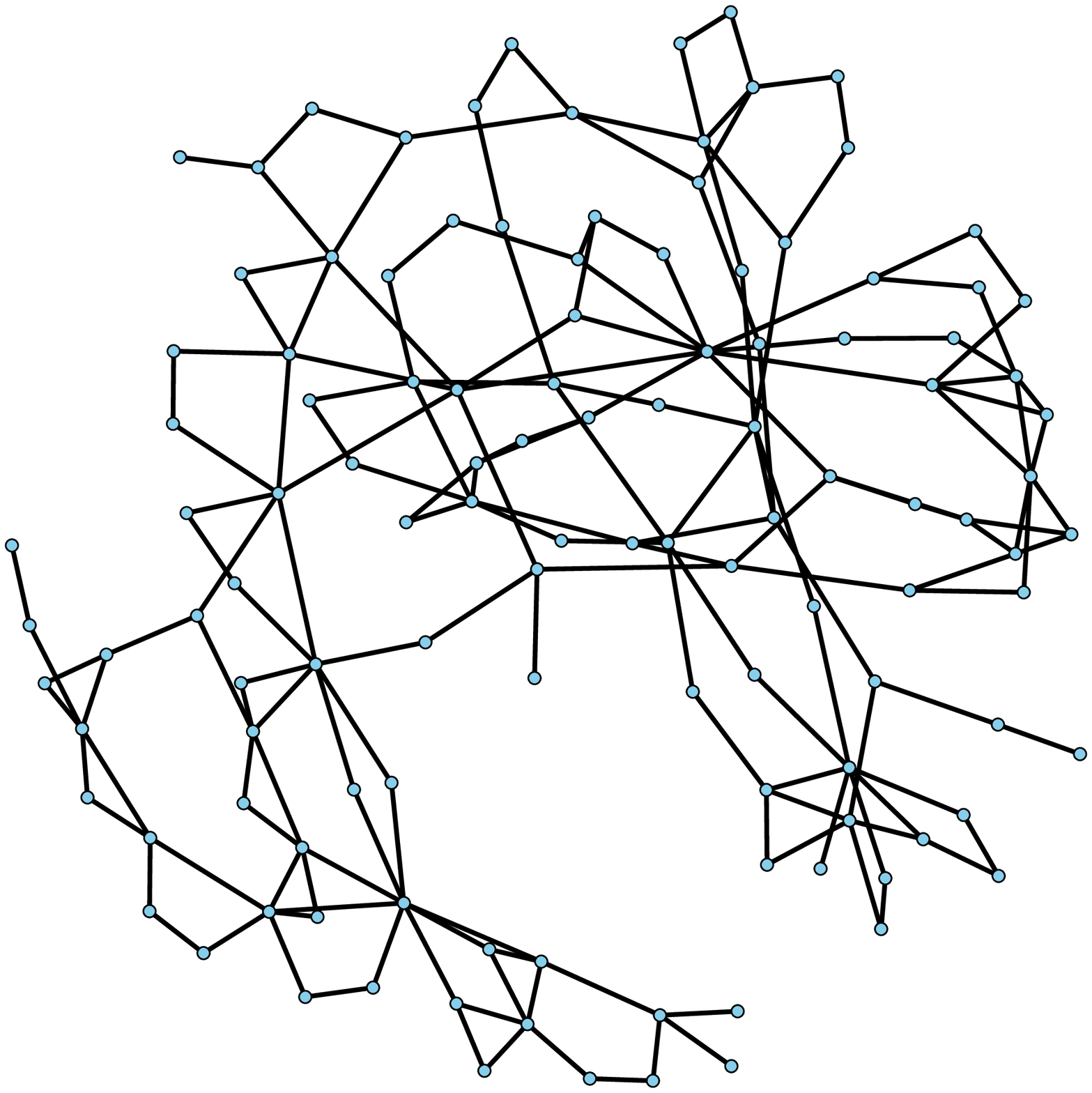}
\captionsetup{format = hang, justification = centering}
\caption{SUD-DBP: $|E(\mathbf{A}^*)| = 179$ \\$\frac{\lambda'}{\mu'} = 0.27542, \gamma'= 28.889$}
\label{fig:118_AN}
\end{subfigure}\hfill

\begin{subfigure}{0.3\linewidth}
\centering
\includegraphics[width=0.95\linewidth]{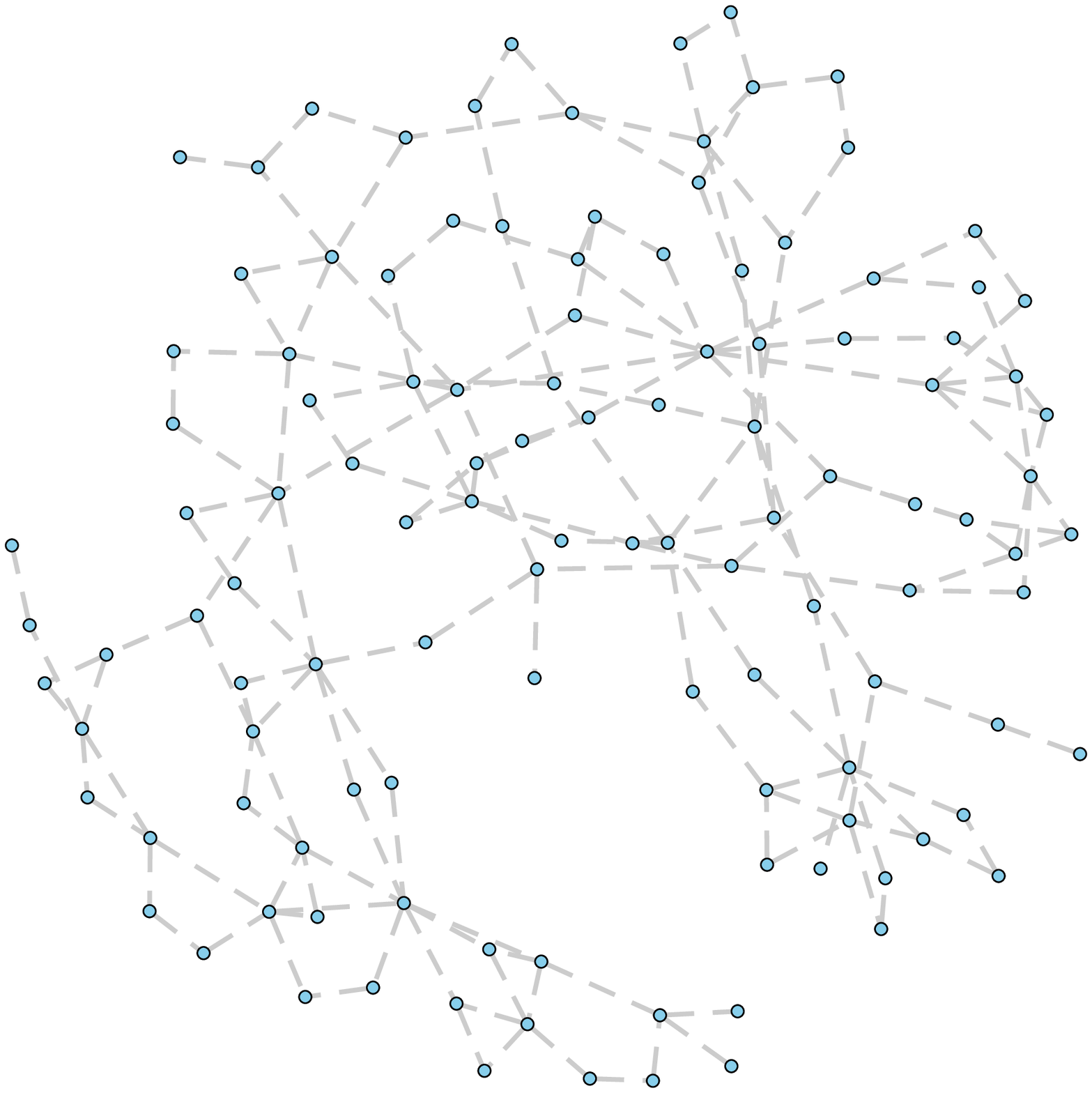}
                \captionsetup{format = hang, justification = centering}
                \caption{TRI-DBP: $|E(\mathbf{A}^*)| = 0$\\ $\frac{\lambda'}{\mu'} = 0.0028, \gamma'= 4.4222$}
                \label{fig:118_A02}
\end{subfigure}\hfill
\begin{subfigure}{0.3\linewidth}
\centering
\includegraphics[width=0.95\linewidth]{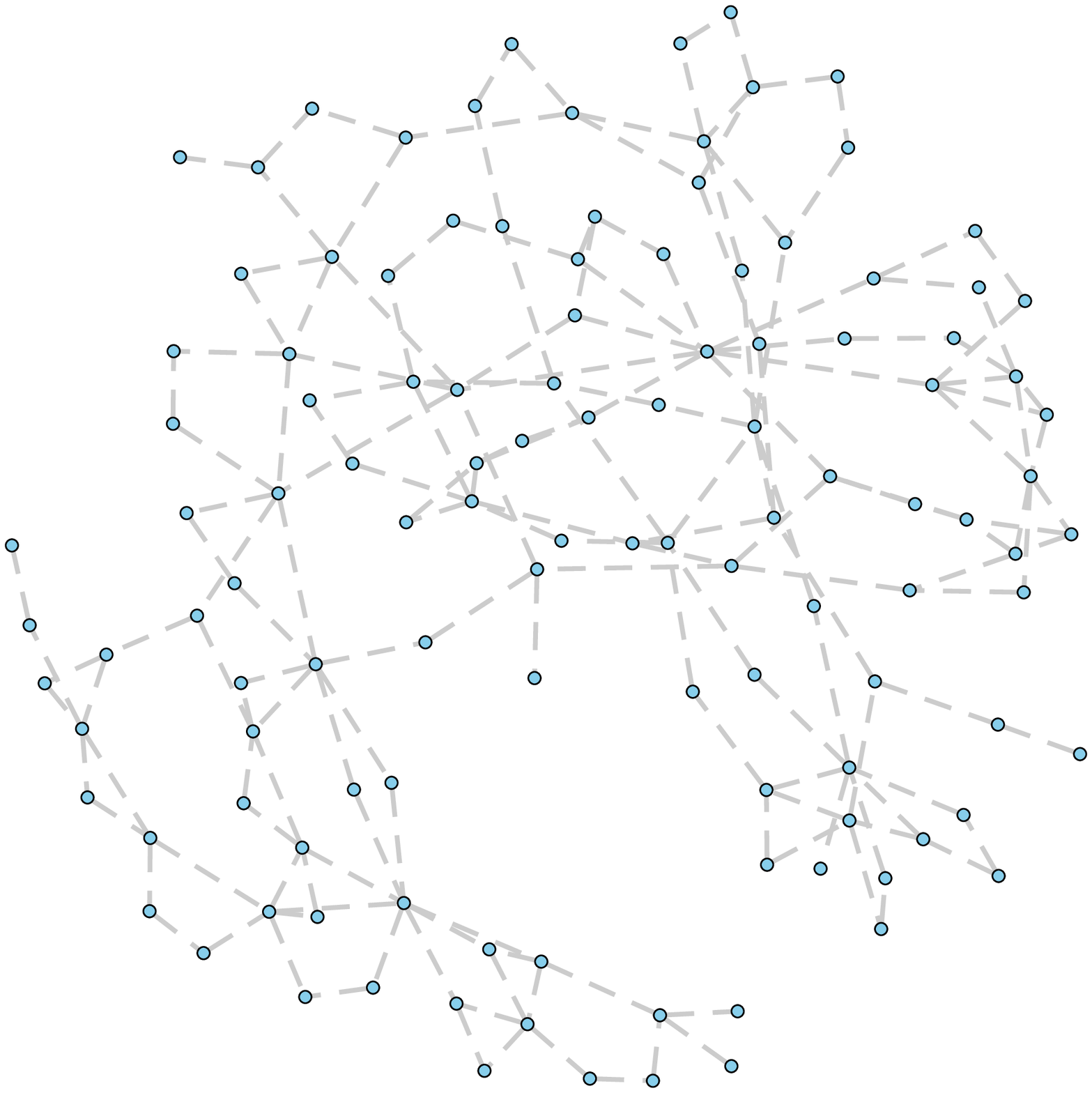}
                   \captionsetup{format = hang, justification = centering}
                \caption{TRI-DBP: $|E(\mathbf{A}^*)| = 0$\\  $\frac{\lambda'}{\mu'} = 0.0167, \gamma'= 1.844$}
                \label{fig:118_A03}
\end{subfigure}\hfill
\begin{subfigure}{0.3\linewidth}
\centering
\includegraphics[width=0.95\linewidth]{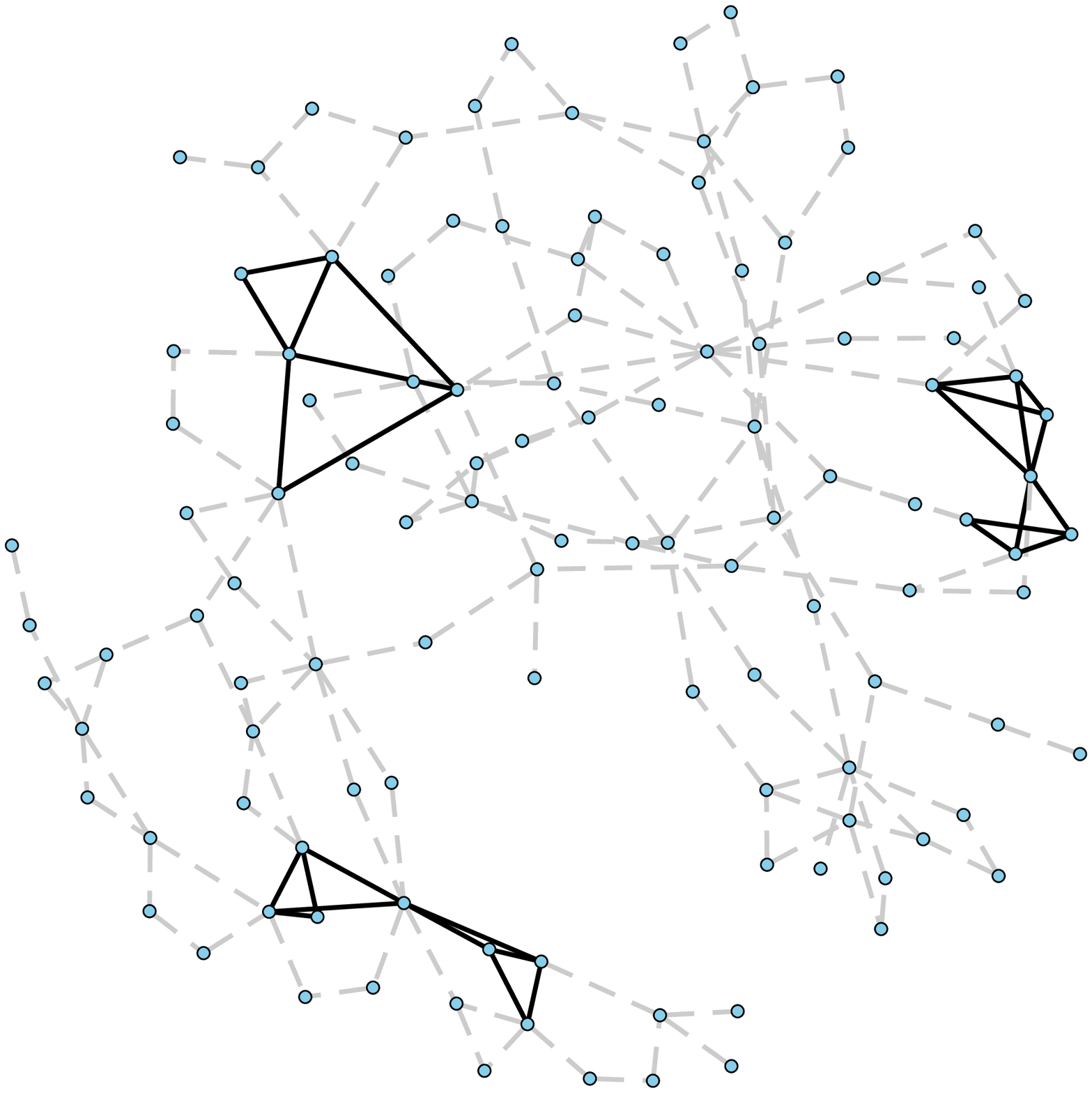}
                  \captionsetup{format = hang, justification = centering}
                \caption{TRI-DBP: $|E(\mathbf{A}^*)| = 28$\\ $\frac{\lambda'}{\mu'} = 0.27542, \gamma'= 28.889$}
                \label{fig:118_TRANresultMatrix_lau0.27542_mu1_g28.889_E28}
\end{subfigure}\hfill

\begin{subfigure}{0.3\linewidth}
\centering
\includegraphics[width=0.95\linewidth]{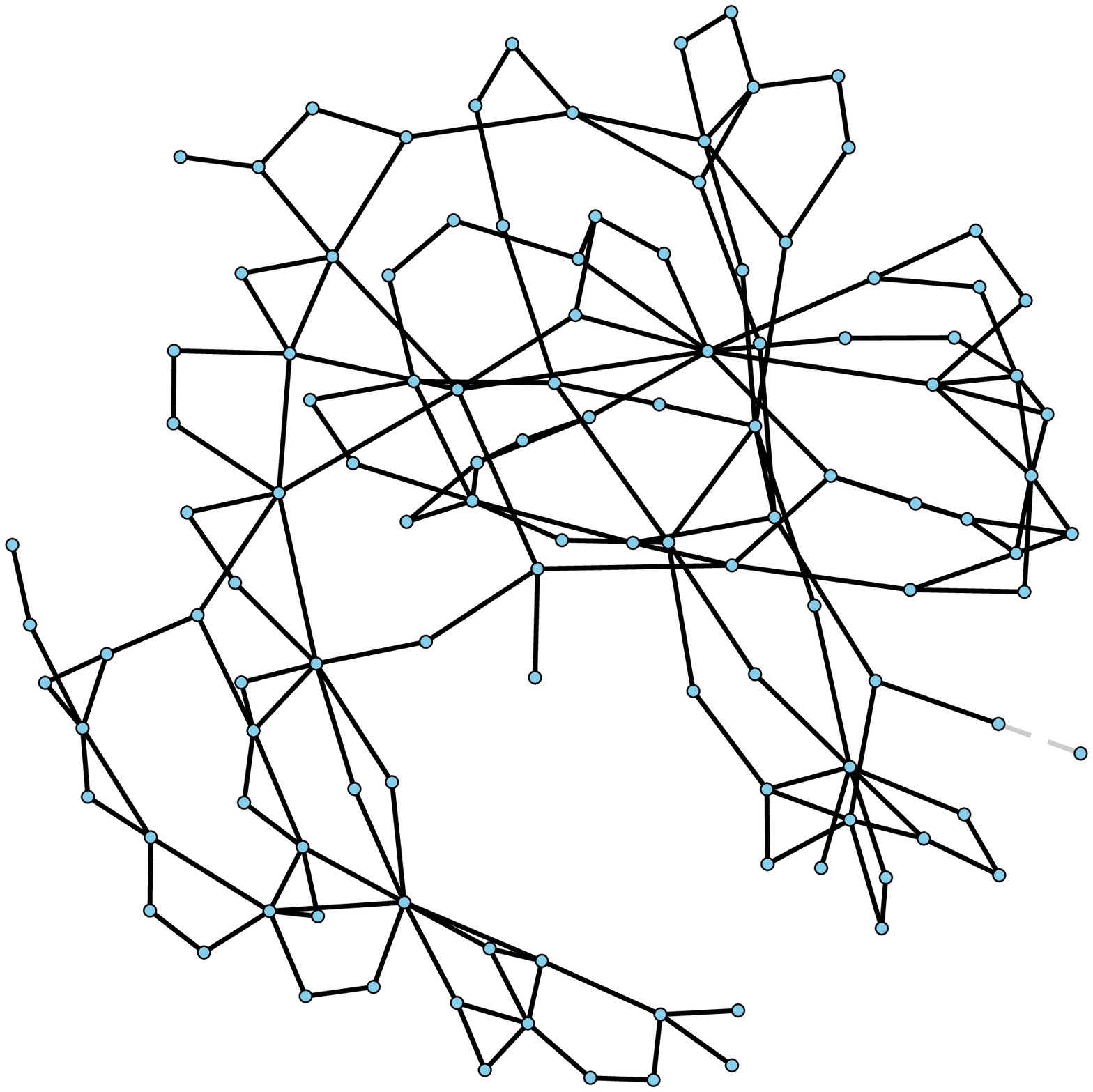}
                  \captionsetup{format = hang, justification = centering}
                \caption{POD-DBP: $|E(\mathbf{A}^*)| = 178$\\ $\frac{\lambda'}{\mu'} = 0.0028, \gamma'= 4.4222$}
                \label{fig:118_PODresultMatrix_lau0.0027826_mu1_g4.4222_E178}
\end{subfigure}\hfill
\begin{subfigure}{0.3\linewidth}
\centering
\includegraphics[width=0.95\linewidth]{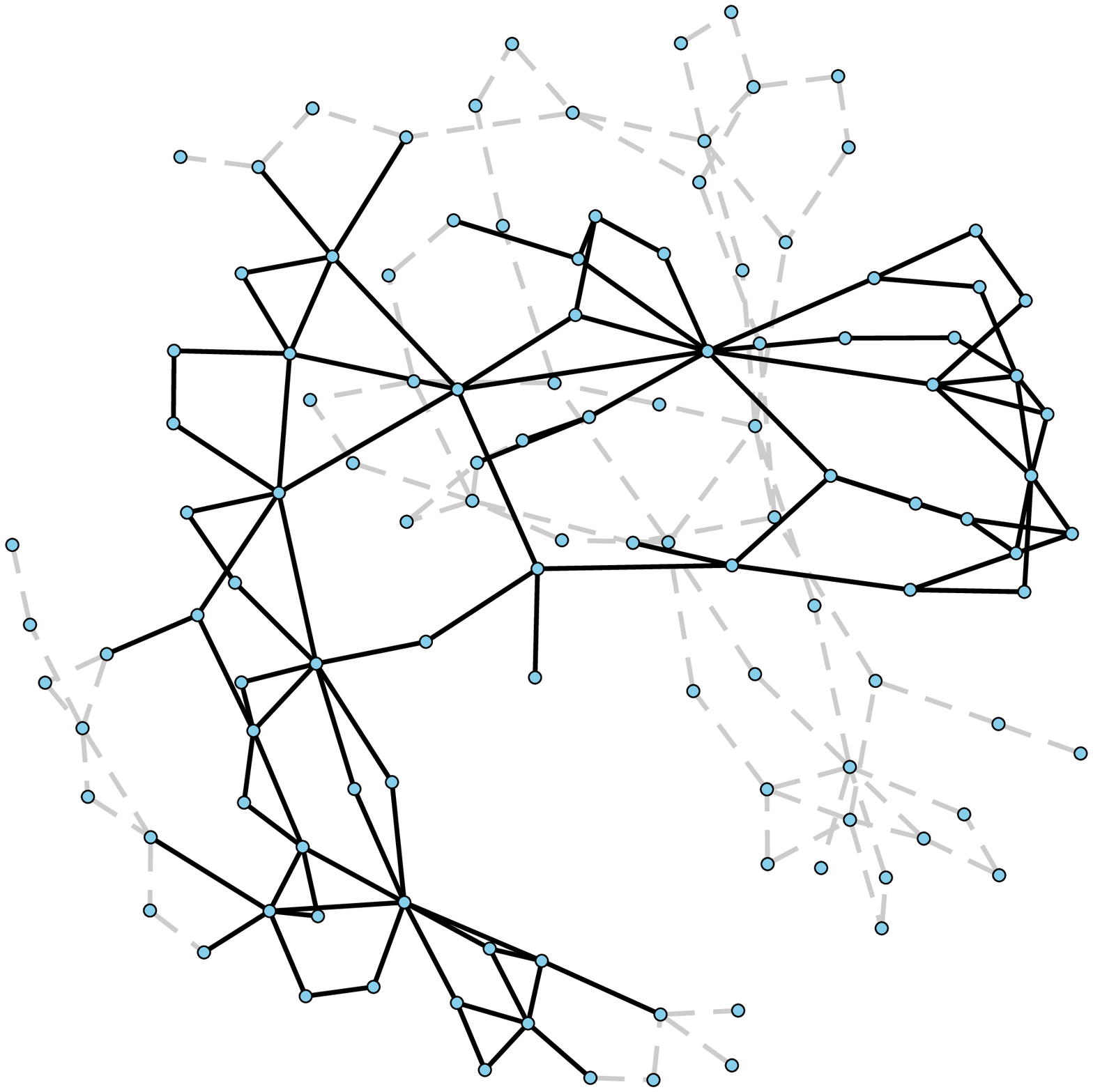}
                 \captionsetup{format = hang, justification = centering}
                \caption{POD-DBP: $|E(\mathbf{A}^*)| = 98$\\ $\frac{\lambda'}{\mu'} = 0.0167, \gamma'= 1.844$}
                \label{fig:118_PODresultMatrix_lau0.016681_mu1_g1.8444_E98}
\end{subfigure}\hfill
\begin{subfigure}{0.3\linewidth}
\centering
\includegraphics[width=0.95\linewidth]{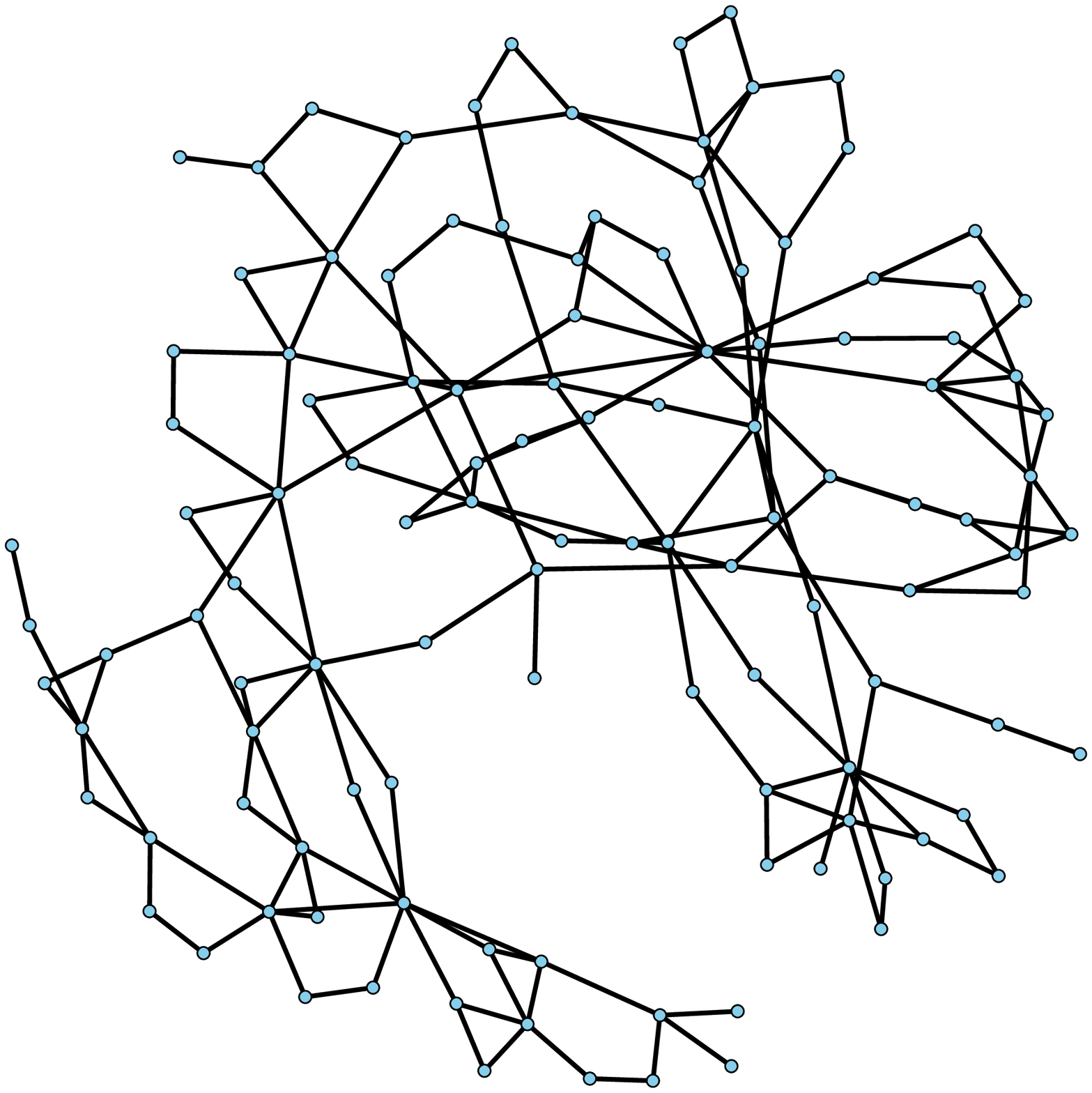}
                \captionsetup{format = hang, justification = centering}
                \caption{SUD-DBP: $|E(\mathbf{A}^*)| = 179$\\ $\frac{\lambda'}{\mu'} = 0.27542, \gamma'= 28.889$}
                \label{fig:118_AN2}
\end{subfigure}\hfill

 \caption{118-bus Test Case Most-Probable Network (solid edges = close, dashed edges = open)}\label{fig:118bus}
\end{figure*}

In this section, we use DBP to model cascading failures of transmission lines on the IEEE 118 Bus Test Case\cite{power:Online}. The network is a portion of the Midwestern US power grid from 1962. It remains one of the standard test cases today. We only use the network topology provided by the test case. The nodes in the network represent the buses and the edges represent transmission lines between the buses. We model edge closures as line failures and edge opening as recovery. During cascading failures like blackouts, transmission line $(a,b)$ is more likely to fail if there are already many failed transmission lines at bus (i.e., node) $a$ and bus $b$. Therefore, it is intuitive to assume that $\gamma' > 1$. On the other hand, spontaneous failures are rare events and since the power grid is well maintained, the recovery rate of failed transmission lines is relatively high; it is natural to assume that  $0 < \frac{\lambda'}{\mu'} \le 1$.

Figure~\ref{fig:118bus} shows the most-probable network with different dynamic parameters $\lambda', \gamma', \mu'$ for SUD-DBP, TRI-DBP, and POD-DBP. When the closure (i.e., failure) rate of an edge is low compared with the opening rate (i.e., recovery), $\mathbf{A}^* = \mathbf{A}_0$ as in Regime I). When the closure (i.e., failure) rate of an edge is high compared with the opening rate (i.e., recovery), $\mathbf{A}^* = \mathbf{A}_{\max}$ as in Regime IV). However, for a subset of parameters, the most-probable networks are \emph{non-degenerate}. This means that a subset of edges are more vulnerable to cascading failures than others. We can also see that the contagion dynamic, captured by the cascade function $f(\mathcal{N}_a, \mathcal{N}_b)$, has a large impact on the susceptibility of the network to failure.

The dynamics of SUD-DBP \eqref{eq:SUDform} put higher probability on network states that minimize the number of closed edges, $|E(\mathbf{A})|$ and maximize the number of induced $P_3$ subgraphs. We can see from Figures~\ref{fig:118_SUDresultMatrix_lau0.0027826_mu1_g4.4222_E9}, that the most-probable network of SUD-DBP consists of closed edge structures resembling many star subgraphs. Central, hub structures are more vulnerable under SUD-DBP dynamics.

The dynamics of TRI-DBP \eqref{eq:TRIform} put high probability on network states that minimize the number of closed edges, $|E(\mathbf{A})|$ but maximize the number of induced $C_3$ subgraphs. From Figure~\ref{fig:118_TRANresultMatrix_lau0.27542_mu1_g28.889_E28}, we can see the the most-probable network consists of triangles. Since the cascade function of TRI-DBP depends on the number of common neighbors, networks with few triangles have the lowest rate of cascading failures.

The dynamics of POD-DBP \eqref{eq:PODform} put high probability on network states that minimize the number of closed edges, $|E(\mathbf{A})|$ but maximize the number of induced $C_3$ and $P_4$ subgraphs. We can see from Figure~\ref{fig:118_PODresultMatrix_lau0.016681_mu1_g1.8444_E98} that the most-probable network consists of more triangles (i.e., $C_3$s) and longer paths than the most-probable network of SUD-DBP. POD-DBP has higher rate of cascade than SUD-DBP and TRI-DBP and therefore more lines are vulnerable to cascading failures for the same parameter values.

%
%


\subsection{Social Network Example}

\begin{figure*}
\begin{subfigure}{0.3\linewidth}
\centering
\includegraphics[width=0.95\linewidth]{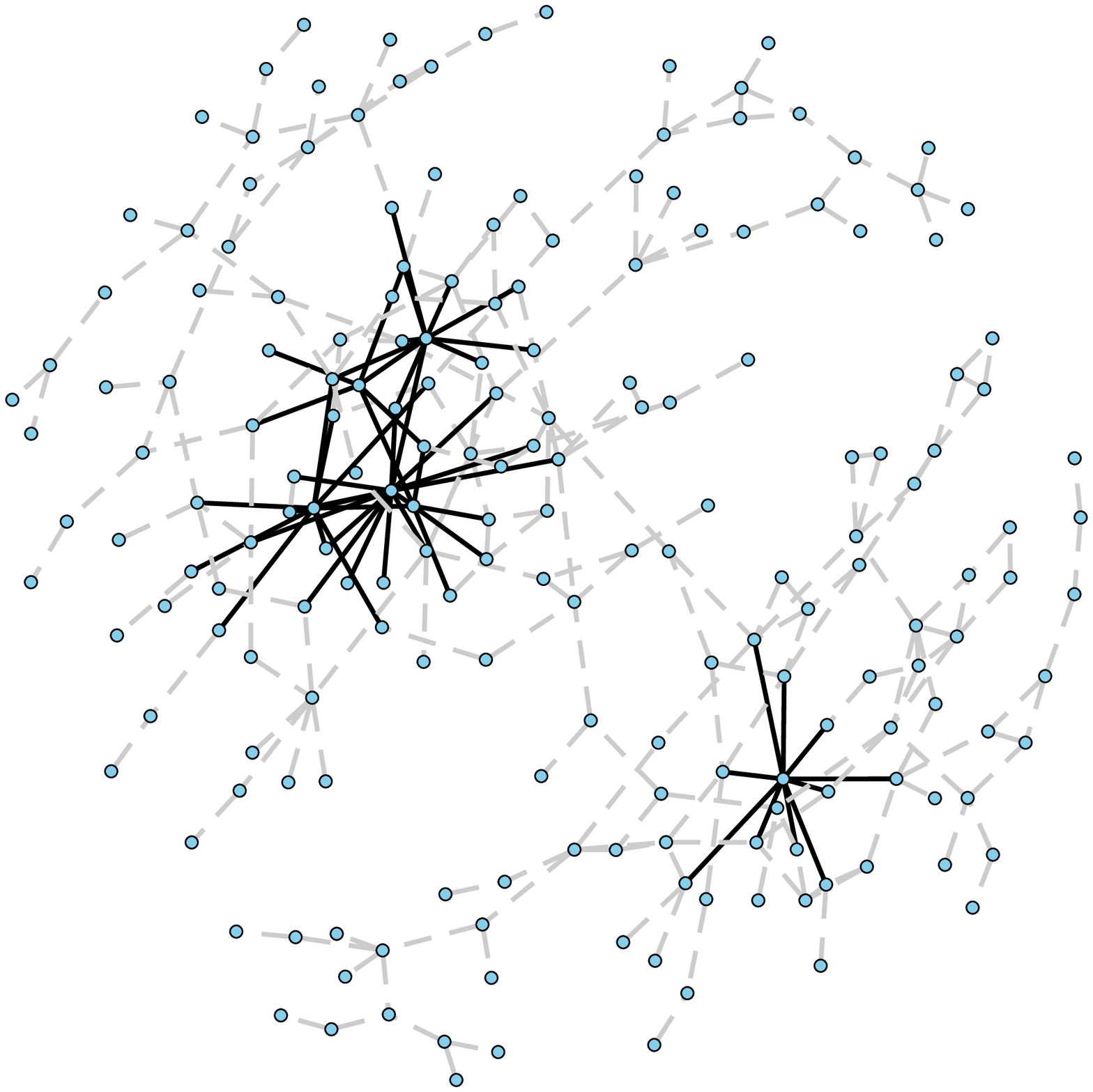}  
               \captionsetup{format = hang, justification = centering}
                \caption{SUD-DBP: $|E(\mathbf{A}^*)| = 54$\\ $\frac{\lambda'}{\mu'} = 0.0028, \gamma'= 4.4222$}
                \label{fig:KPP_SUDresultMatrix_lau0.0027826_mu1_g4.4222_E54}
\end{subfigure}\hfill
\begin{subfigure}{0.3\linewidth}
\centering
\includegraphics[width=0.95\linewidth]{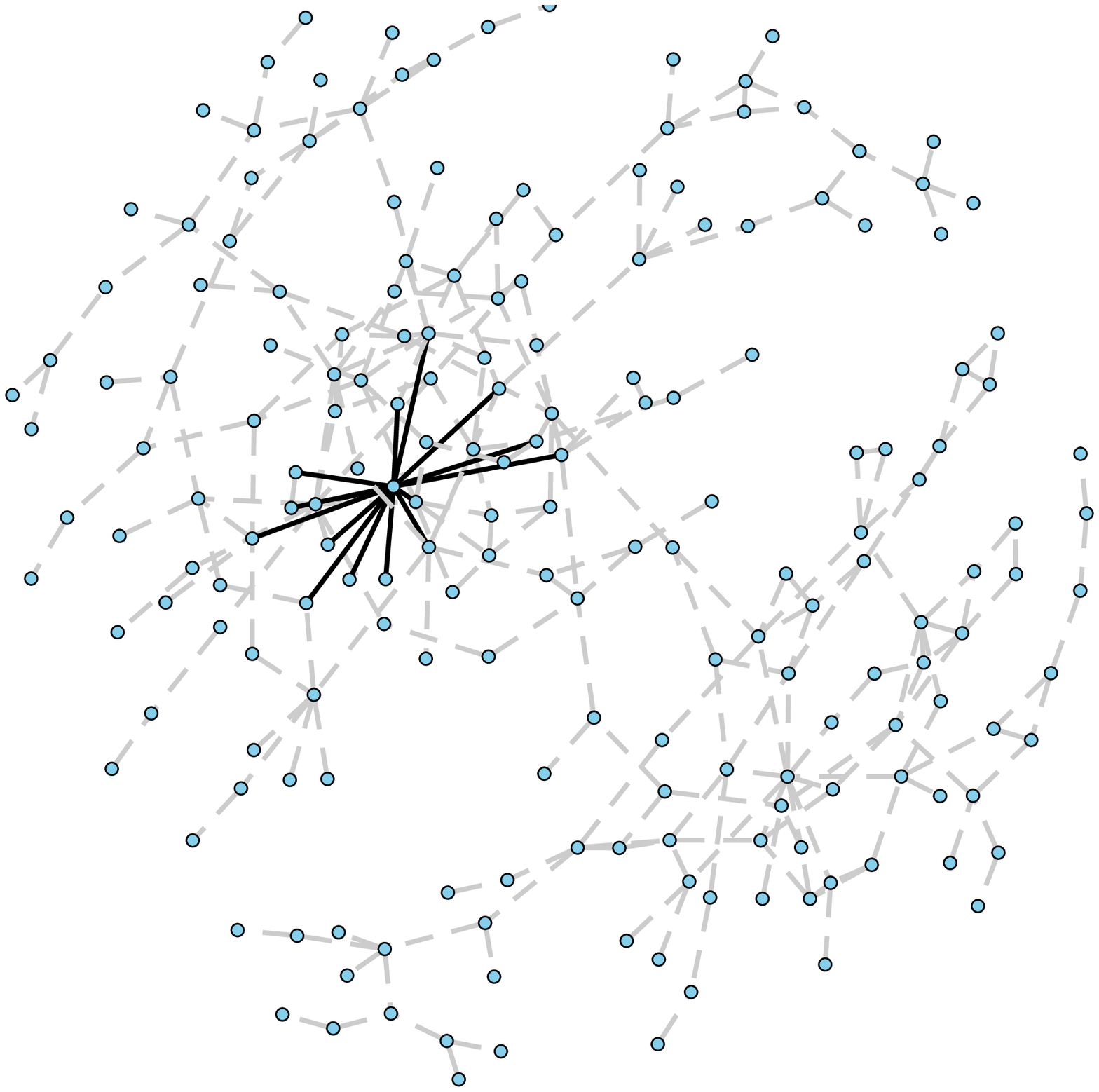}
                \captionsetup{format = hang, justification = centering}
                \caption{SUD-DBP: $|E(\mathbf{A}^*)| = 15$\\ $\frac{\lambda'}{\mu'} = 0.0167, \gamma'= 1.844$}
                \label{fig:KPP_SUDresultMatrix_lau0.016681_mu1_g1.8444_E15}
\end{subfigure}\hfill
\begin{subfigure}{0.3\linewidth}
\centering
\includegraphics[width=0.95\linewidth]{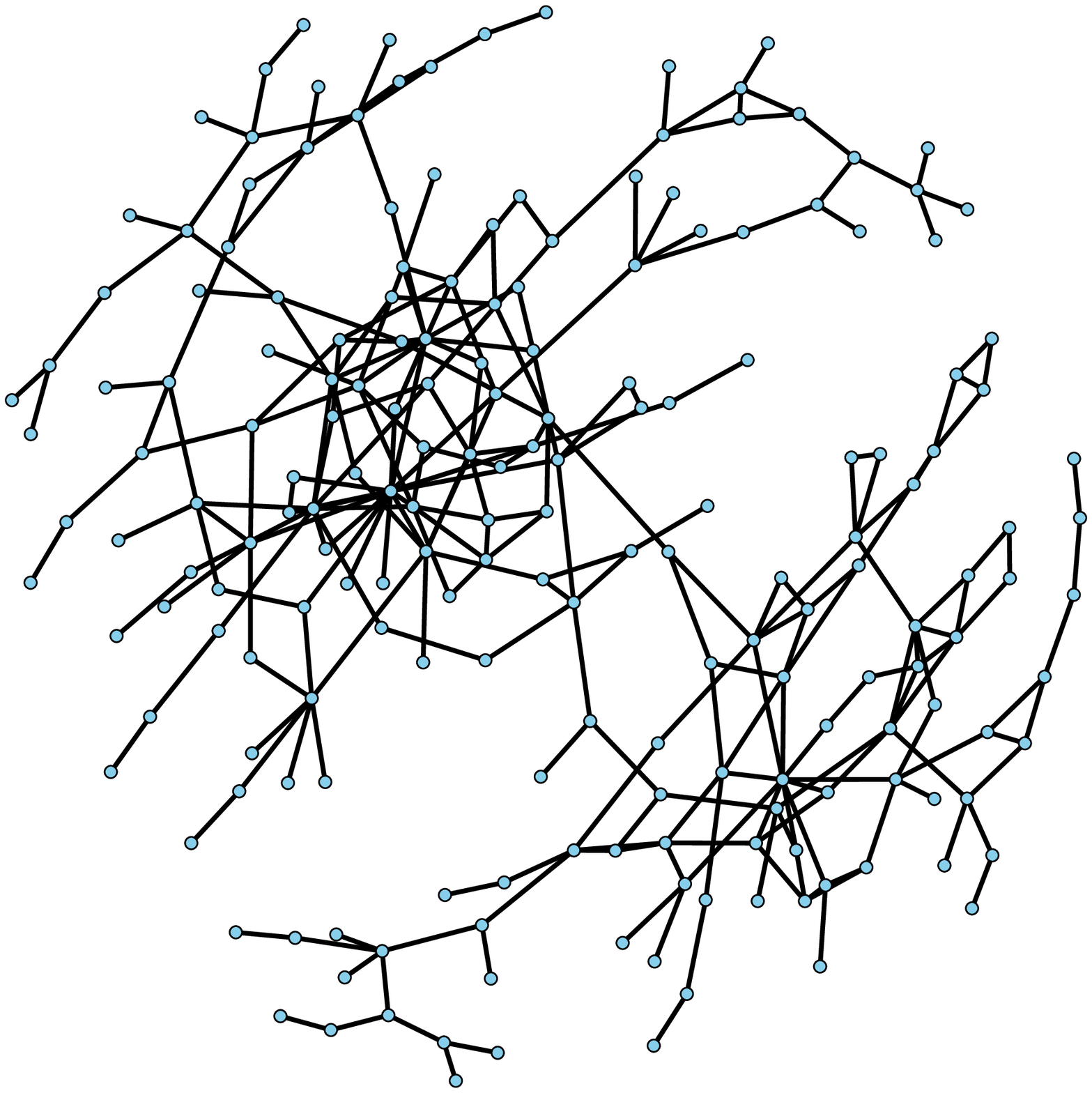}
                \captionsetup{format = hang, justification = centering}
                \caption{SUD-DBP: $|E(\mathbf{A}^*)| = 273$\\ $\frac{\lambda'}{\mu'} = 0.27542, \gamma'= 28.889$}
                \label{fig:KPP_AN}
\end{subfigure}\hfill

\begin{subfigure}{0.3\linewidth}
\centering
\includegraphics[width=0.95\linewidth]{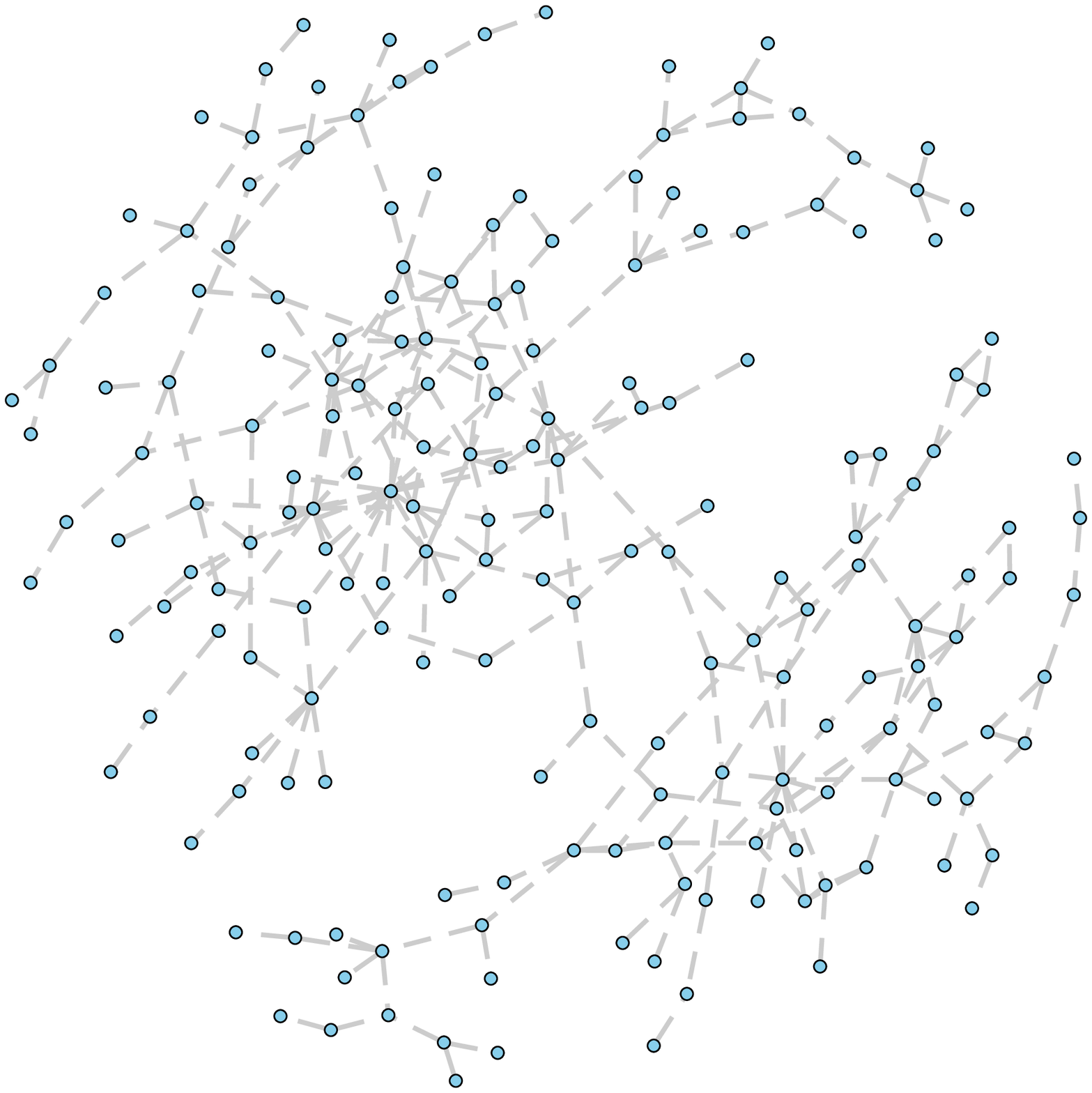}
                \captionsetup{format = hang, justification = centering}
                \caption{TRI-DBP: $|E(\mathbf{A}^*)| = 0$\\ $\frac{\lambda'}{\mu'} = 0.0028, \gamma'= 4.4222$}
                \label{fig:KPP_A0}
\end{subfigure}\hfill
\begin{subfigure}{0.3\linewidth}
\centering
\includegraphics[width=0.95\linewidth]{8e.eps}
                \captionsetup{format = hang, justification = centering}
                \caption{TRI-DBP: $|E(\mathbf{A}^*)| = 0$\\ $\frac{\lambda'}{\mu'} = 0.0167, \gamma'= 1.844$}
                \label{fig:KPP_A02}
\end{subfigure}\hfill
\begin{subfigure}{0.3\linewidth}
\centering
\includegraphics[width=0.95\linewidth]{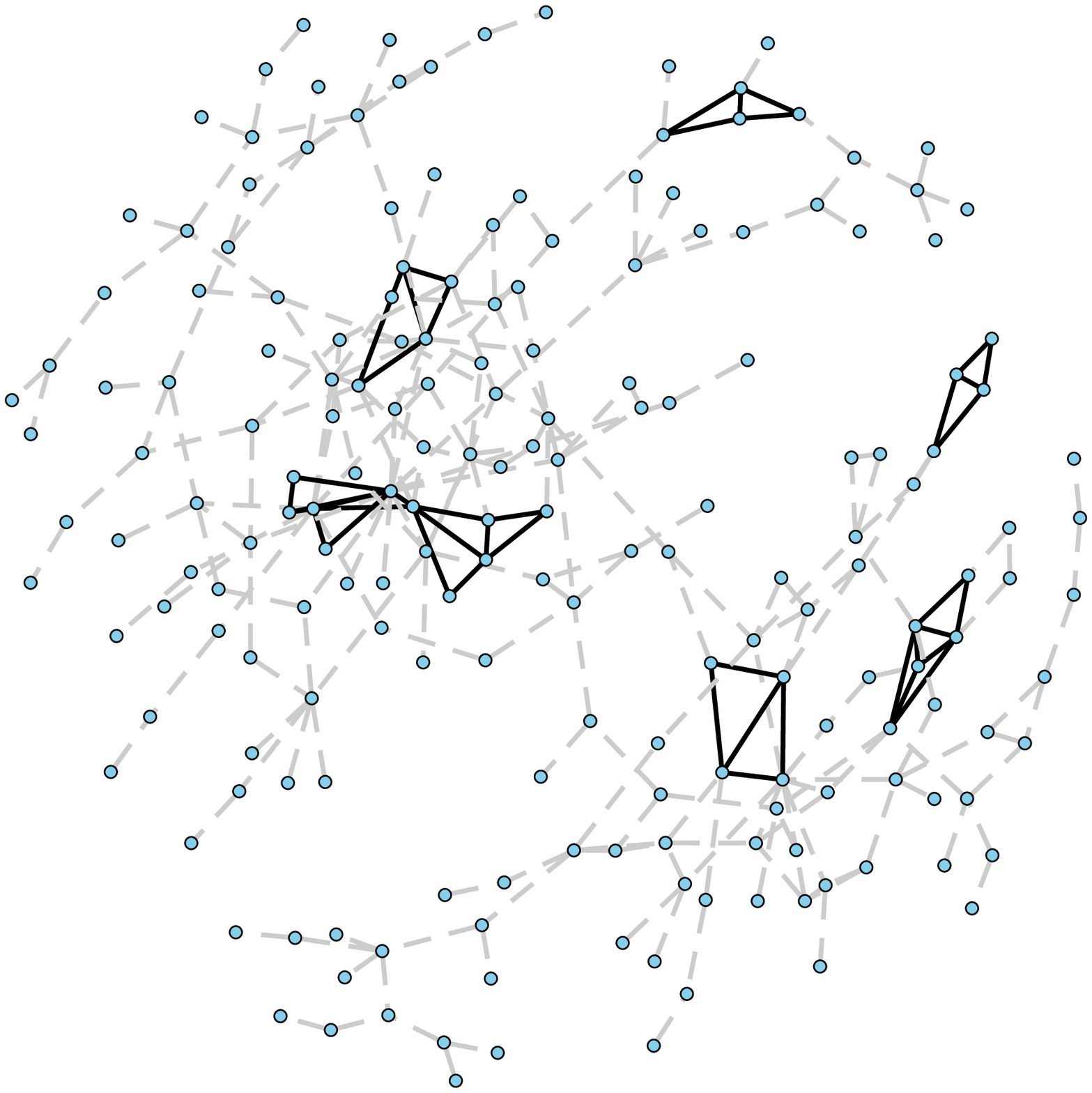}
                \captionsetup{format = hang, justification = centering}
                \caption{TRI-DBP: $|E(\mathbf{A}^*)| = 44$\\ $\frac{\lambda'}{\mu'} = 0.27542, \gamma'= 28.889$}
                \label{fig:KPP_TRANresultMatrix_lau0.27542_mu1_g28.889_E44}
\end{subfigure}\hfill

\begin{subfigure}{0.3\linewidth}
\centering
\includegraphics[width=0.95\linewidth]{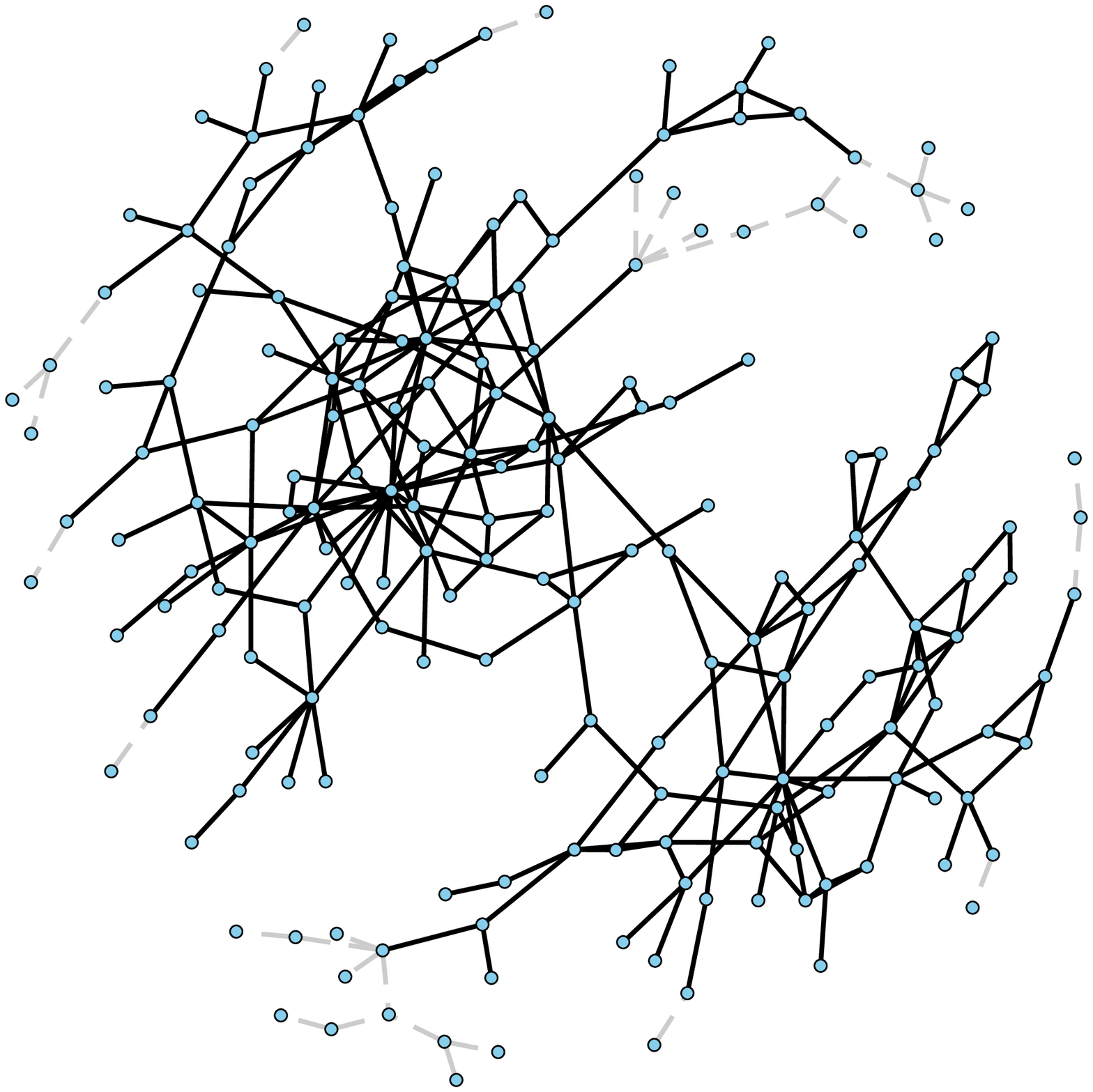}
                 \captionsetup{format = hang, justification = centering}
                \caption{POD-DBP: $|E(\mathbf{A}^*)| = 241$\\ $\frac{\lambda'}{\mu'} = 0.0028, \gamma'= 4.4222$}
                \label{fig:KPP_PODresultMatrix_lau0.0027826_mu1_g4.4222_E241}
\end{subfigure}\hfill
\begin{subfigure}{0.3\linewidth}
\centering
\includegraphics[width=0.95\linewidth]{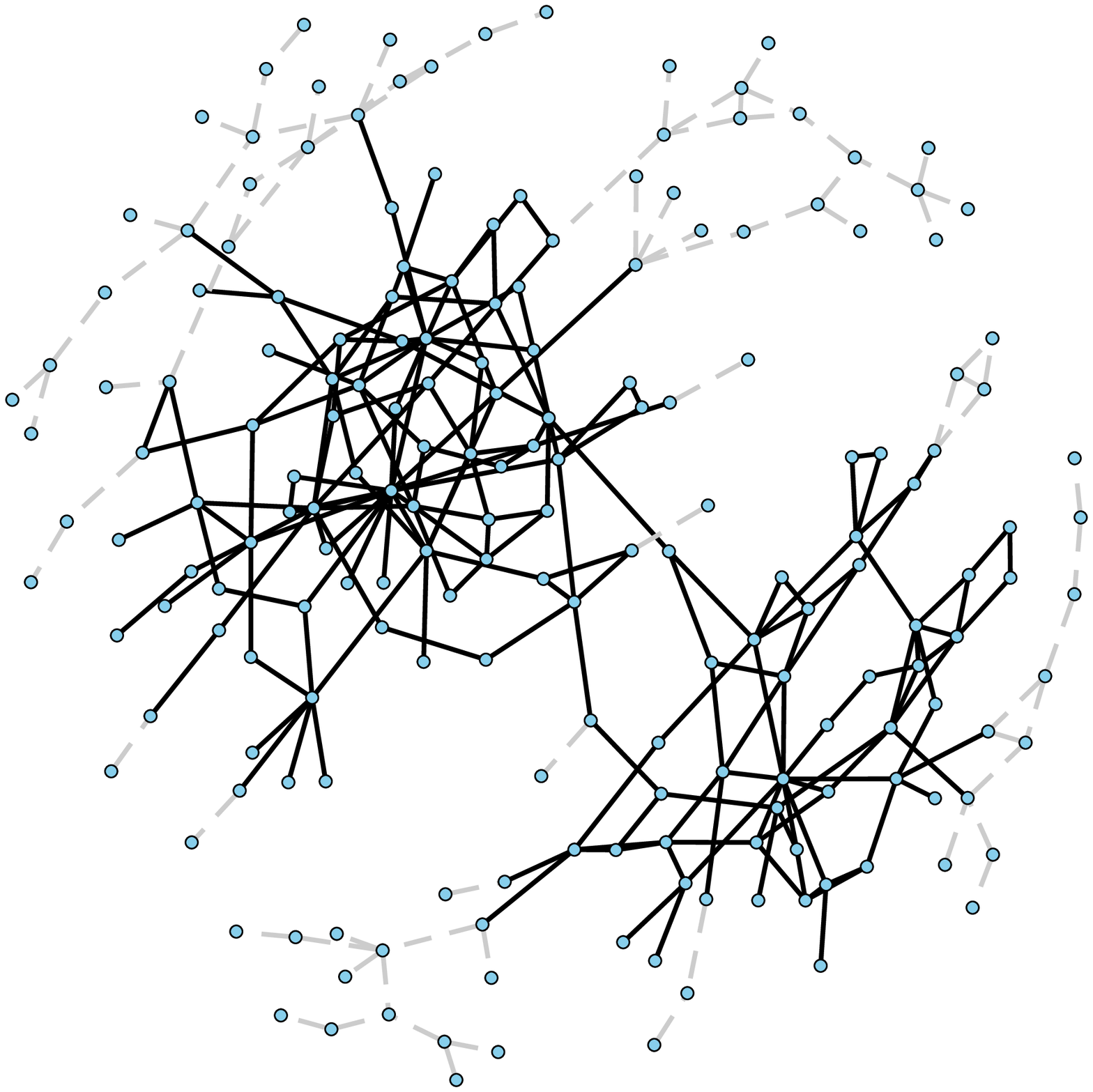}
                 \captionsetup{format = hang, justification = centering}
                \caption{POD-DBP: $|E(\mathbf{A}^*)| = 194$\\ $\frac{\lambda'}{\mu'} = 0.0167, \gamma'= 1.844$}
                \label{fig:KPP_PODresultMatrix_lau0.016681_mu1_g1.8444_E194}
\end{subfigure}\hfill
\begin{subfigure}{0.3\linewidth}
\centering
\includegraphics[width=0.95\linewidth]{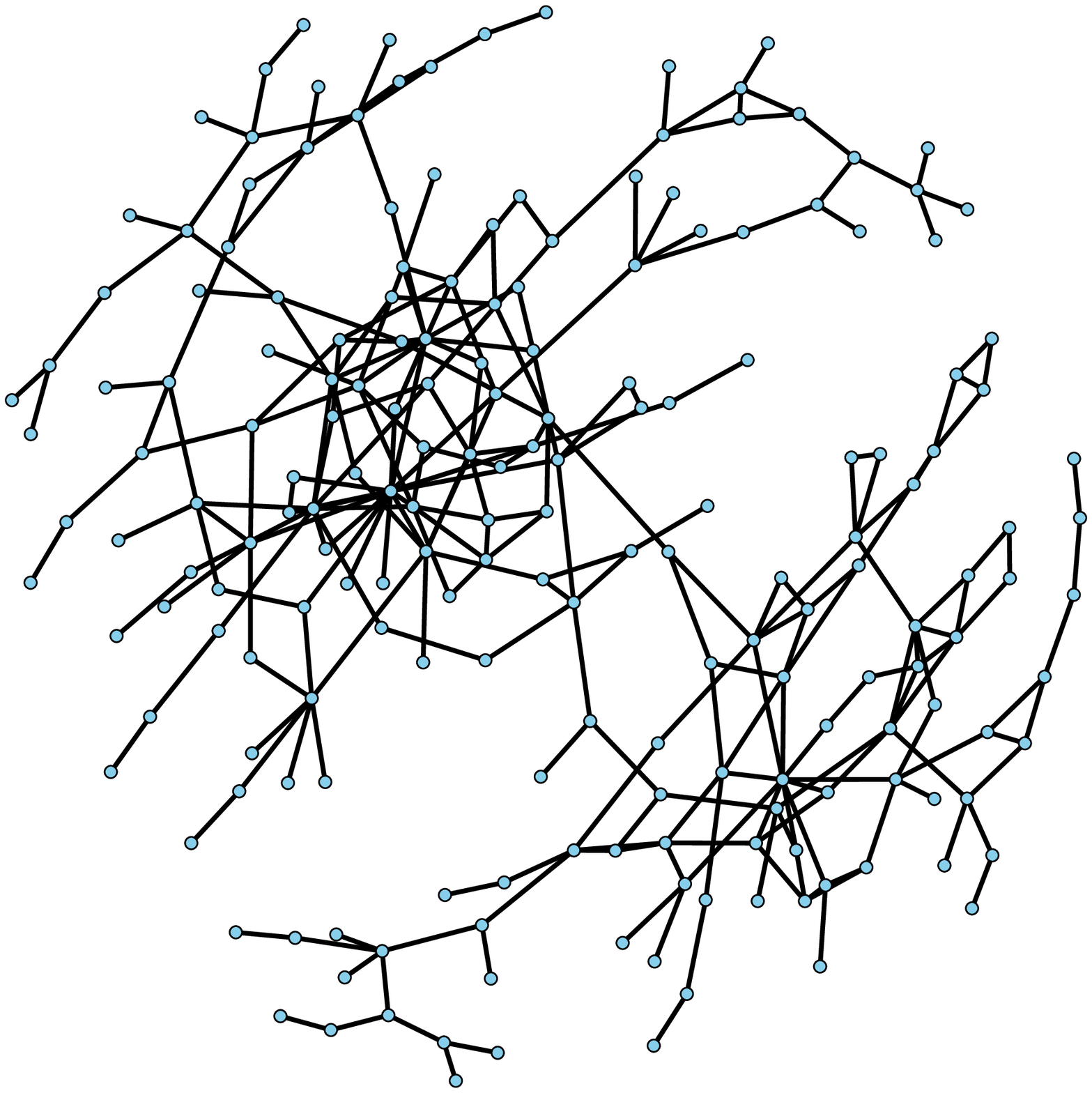}
                 \captionsetup{format = hang, justification = centering}
                \caption{POD-DBP: $|E(\mathbf{A})| = 273$ \\ $\frac{\lambda'}{\mu'} = 0.27542, \gamma'= 28.889$}
                \label{fig:KPP_AN2}
\end{subfigure}\hfill

 \captionsetup{format = hang}
 \caption{198-node Social Network Most-Probable Network (solid edges = closed, dashed edges = open)}\label{fig:KPP}
\end{figure*}

We also used DBP to model evolving social ties on a 198-node social network of drug users in Hartford, CT \cite{weeks2002social}. The nodes in the network represent individual drug users. Edge closure corresponds to formation or reestablishment of social ties, and edge opening corresponds to dissipation of social ties. In this case, edges that are closed in the most-probable networks are the social ties that are often established. They are the stronger social ties between individuals.

Assuming that $\gamma' > 1$ means that a relationship $(a,b)$ is more likely to form if agent $a$ and agent $b$ already have many other relationships. In particular, TRI-DBP assumes that a relationship $(a,b)$ is more likely to form if agent $a$ and agent $b$ already have many friends in common. This is based on the theory of triadic closure from social networks \cite{brandes2013studying}. On the other hand, spontaneous friendships are possible and relationships can dissipate over time. Therefore, it is natural to assume nonzero $\lambda'$ and $\mu'$ and that $0 < \frac{\lambda'}{\mu'} \le 1$.

Figure~\ref{fig:KPP} shows the most-probable network with different dynamic parameters $\lambda', \gamma', \mu'$ for SUD-DBP, TRI-DBP, and POD-DBP. We can see from Figure~\ref{fig:KPP_TRANresultMatrix_lau0.27542_mu1_g28.889_E44} and Figure~\ref{fig:KPP_PODresultMatrix_lau0.016681_mu1_g1.8444_E194} that, for some range of parameter values, social ties in triadic closures are stronger assuming TRI-DBP and POD-DBP cascade functions. On the other hand, social ties of highly connected individuals are stronger as in SUD-DBP.

\section{Regime III) Cascading Prevention}\label{sec:regimeIIIedge}

In Regime III) Cascading Prevention: $\frac{\lambda'}{\mu'} > 1$ and $0 < \gamma' \le 1$, the structure-free term is driven by edge closure and the structure-dependent term is driven by edge opening. The dynamics of Regime III) Cascading Prevention is the opposite of Regime II) Cascading Failures. The average time an edge is open increases  with increasing number of closed edges on its end nodes; diffusion effects, instead of driving cascading failures, preventsedge closures. Therefore, this regime is called Regime III) Cascading Prevention. For SUD-DBP, TRI-DBP, and POD-DBP, we expect the solution space of the Most-Probable Network Problem \eqref{eq:astar} to exhibit phase transition depending on if edge opening or edge closure dominates.

However, since $0 < \gamma' \le 1$ in Regime III), the Most-Probable Network Problem can not be solved using submodularity according to Theorem~\ref{theorem:networksubmodular}. However, we will prove in the next section that we can still solve for the most-probable network in polynomial-time for a sub-range of parameter values in Regime III) for SUD-DBP. In this sub regime, the most-probable network corresponds to the maximum matching (see Definition \ref{def:matching}). For TRI-DBP, the most-probable network corresponds to triangle-free graphs. For POD-DBP, the most-probable network corresponds to the maximum star matching (see Definition~\ref{def:starmatching}).

\subsection{SUD-DBP and Maximum Matching}

The dynamics of SUD-DBP in regime III) put high probability on network states that maximize the number of closed edges, $|E(\mathbf{A})|$, and minimize the number of induced $P_3$ subgraphs. This means avoiding forming paths of length $2$ and allowing for only paths of length $1$. As a result, for sub-range of parameters, $\frac{\lambda'}{\mu'}, \gamma'$, we can prove that the set of close edges in the most-probable network is a maximum matching (see Definition~\ref{def:matching}). It is known that the maximum matching can be found in polynomial time for arbitrary, undirected graphs \cite{micali1980v, harvey2009algebraic}. Maximum matching may not be unique.

\begin{theorem}\label{theorem:maxmatching}
In Regime III), if $\lambda'\gamma' < \mu'$, then the most-probable network, 
$\mathbf{A}^* = \{ \mathbf{A} \in \mathcal{A} : g(E(\mathbf{A})) = 0, |E(\mathbf{A})| \text{ is maximum}\}$, where $g(\mathbf{A}) $ is the number of induced $P_3$ subgraphs. This is equivalent to stating that $E(\mathbf{A}^*)$ is a maximum matching (see Definition~\ref{def:matching}).

\end{theorem}

\begin{proof}
We prove the theorem by contradiction. 
Let $\mathcal{A}^*$ denote the set of network states whose set of closed edges are maximum matching. 
\[
\mathcal{A}^* = \{ \mathbf{A} \in \mathcal{A} : E(\mathbf{A}) \text{ is maximum matching}\}.
\]
By the definition of matching, the number of $P_3$ subgraphs induced by any network state $\mathbf{A}^* \in \mathcal{A}^*$, $g(E(\mathbf{A}^*)) = 0$. Suppose that the most-probable network is $\mathbf{A}' \in \mathcal{A} \setminus \mathcal{A}^*$. Then there are two possibilities for $\mathbf{A}'$:
\begin{enumerate}
\item $\mathbf{A}'$ is the network state such that $E(\mathbf{A}')$ is a matching but it is not the maximum matching.
\item $\mathbf{A}'$ is the network state such that $E(\mathbf{A}')$ is not a matching. 
\end{enumerate}

\begin{description}
\item [Case 1) ] \hfill \\
This implies that $|E(\mathbf{A}')| < |E(\mathbf{A}^*)|$ for any $\mathbf{A}^* \in \mathcal{A}^*$. Since $\frac{\lambda'}{\mu'} > 1$, then by the equilibrium distribution \eqref{eq:sumeq}, $\pi(\mathbf{A}') < \pi(\mathbf{A}^*)$. Therefore, $\mathbf{A}'$ can not be the most-probable network.

\item [Case 2)  ] \hfill \\
This implies that $g(E(\mathbf{A}')) > 0$. For any $\mathbf{A}^* \in \mathcal{A}^*$, there are two possibilities 

\begin{enumerate}
\item $|E(\mathbf{A}')| \le |E(\mathbf{A}^*)|$  \hfill \\
If $|E(\mathbf{A}')| \le |E(\mathbf{A}^*)|$, then 
\[
\left(\frac{\lambda'}{\mu'}\right)^{|E(\mathbf{A}')|}\gamma'^{g(E(\mathbf{A}'))} < \left( \frac{\lambda'}{\mu'}\right)^{|E(\mathbf{A}^*)|}\gamma'^{g(E(\mathbf{A}^*))}.
\]
Therefore, $\mathbf{A}'$ can not be the most-probable network.

\item $|E(\mathbf{A}')| > |E(\mathbf{A}^*)|$.

Let $\mathcal{A}'_1$ denote the set of network states such that for any $\mathbf{A}' \in \mathcal{A}'_1$, $|E(\mathbf{A}')| =  |E(\mathbf{A}^*)|+1$. This also implies that the number of $P_3$ subgraphs induced by $E(\mathbf{A}')$, $g(E(\mathbf{A})) \ge 1$:
\[
\mathcal{A}'_1 = \{ \mathbf{A} \in \mathcal{A}: g(E(\mathbf{A})) \ge 1, |E(\mathbf{A})| = |E(\mathbf{A}^*)| + 1\}
\]
Similarly, let $\mathcal{A}'_2$ denote the set of network states such that for any $\mathbf{A}' \in \mathcal{A}'_2$, $|E(\mathbf{A}')| =  |E(\mathbf{A}^*)|+2$: 
\[
\mathcal{A}'_2 = \{ \mathbf{A} \in \mathcal{A}: g(E(\mathbf{A})) \ge 2, |E(\mathbf{A})| = |E(\mathbf{A}^*)| + 2\}.
\]
We can define additional sets such as $\mathcal{A}'_3$ $\mathcal{A}'_4, \ldots, \mathcal{A}'_N$. Realize that $g(E(\mathbf{A}'_2)) > g(E(\mathbf{A}'_1))$ for any $\mathbf{A}'_1\in \mathcal{A}'_1$ and $\mathbf{A}'_2\in \mathcal{A}'_2$, $g(E(\mathbf{A}'_3)) > g(E(\mathbf{A}'_2))$ for any $\mathbf{A}'_2\in \mathcal{A}'_2$ and $\mathbf{A}'_3\in \mathcal{A}'_3$. Similarly $g(E(\mathbf{A}'_N)) > g(E(\mathbf{A}'_{N-1}))$ for any $\mathbf{A}'_{N-1} \in \mathcal{A}'_{N-1}$ and $\mathbf{A}'_{N} \in \mathcal{A}'_{N}$.

Since $0 < \gamma' \le 1$, the network state with the maximum equilibrium probability in set $\mathcal{A}'_1$ has the probability 
\[
\pi(\overline{\mathbf{A}'_1})= \frac{1}{Z}\left(\frac{\lambda'}{\mu'}\right)^{|E(\mathbf{A}^*)| +1}\gamma'
\]
and the network state with the maximum equilibrium probability in set $\mathcal{A}'_2$ has the probability 
\[
\pi(\overline{\mathbf{A}'_2})= \frac{1}{Z}\left(\frac{\lambda'}{\mu'}\right)^{|E(\mathbf{A}^*)| +2}\gamma'^{2}.
\]
The additional condition $\lambda'\gamma' < \mu'$ implies that $\pi(\overline{\mathbf{A}'_1})> \pi(\overline{\mathbf{A}'_2})$. Similar argument will show that $\pi(\overline{\mathbf{A}'_2})  > \pi(\overline{\mathbf{A}'_3})$, and that $\pi(\mathbf{A}'_{N-1}) > \pi(\mathbf{A}'_N)$.

Since 
\[
\pi(\mathbf{A}^*) = \frac{1}{Z}\left(\frac{\lambda'}{\mu'}\right)^{k},
\]
the network state $\mathbf{A}'$ can not be the most-probable network as $\pi(\mathbf{A}') \le \pi(\overline{\mathbf{A}'_1})$ by definition.

\end{enumerate}

\end{description}

\end{proof}


The dynamics of SUD-DBP in regime III) put high probability on network states that maximize the number of closed edges, $|E(\mathbf{A})|$, and minimize the number of induced $P_3$ subgraphs. This means avoiding forming paths of length $2$ and allowing for only paths of length $1$. As a result, for a sub-range of parameters, $\frac{\lambda'}{\mu'}, \gamma'$, we can prove that the set of close edges in the most-probable network is a maximum matching (see Definition~\ref{def:matching}). It is known that the maximum matching can be found in polynomial time for arbitrary, undirected graphs \cite{micali1980v, harvey2009algebraic}. Note that the maximum matching may not be unique.

\subsection{TRI-DBP and Triangle Free Graphs}
The dynamics of TRI-DBP in regime III) put higher probability on network states that maximize the number of closed edges, $|E(\mathbf{A})|$ and minimize the number of induced  $C_3$ subgraphs. This means avoiding forming cycles of length $3$ (i.e., triangles). Consequently, the most-probable configuration will be biased toward the set of closed edges that induces triangle-free graphs. 

\begin{theorem}\label{theorem:maxstarmatching}
In Regime III), if $\lambda'\gamma' < \mu'$, then the most-probable network, 
$\mathbf{A}^* = \{ \mathbf{A} \in \mathcal{A} : g(E(\mathbf{A})) = 0, |E(\mathbf{A})| \text{ is maximum}\}$, where $g(\mathbf{A})$ is the number of induced $C_3$ subgraphs. This is equivalent to stating that $\mathbf{A}^*$ is the largest possible triangle-free subgraph in the maximal network.
\end{theorem}

The proof follows the same argument as in Theorem~\ref{theorem:maxmatching}.

\subsection{POD-DBP and Maximum Star Matching}
The dynamics of POD-DBP in regime III) put higher probability on network states that maximize the number of closed edges, $|E(\mathbf{A})|$ and minimize the number of induced  $C_3$ and $P_4$ subgraphs. This means avoiding forming paths of length $3$ and cycles of length $3$ and allowing paths of length $1$ and $2$. Consequently, the most-probable configuration will be biased toward the set of closed edges that maximizes the number of induced paths of length 2. 

\begin{theorem}\label{theorem:maxstarmatching}
In Regime III), if $\lambda'\gamma' < \mu'$, then the most-probable network, 
$\mathbf{A}^* = \{ \mathbf{A} \in \mathcal{A} : g(E(\mathbf{A})) = 0, |E(\mathbf{A})| \text{ is maximum}\}$, where $g(\mathbf{A})$ is the number of induced $C_3$ and $P_4$ subgraphs. This is equivalent to stating that $E(\mathbf{A}^*)$ is a maximum star matching (see Definition~\ref{def:starmatching}).

\end{theorem}

The proof follows the same argument as in Theorem~\ref{theorem:maxmatching}.

\section{Conclusion}\label{sec:conclusion}

We presented the Dynamic Bond Percolation process, a stochastic, edge-centric network process. DBP assumed that an edge in the maximal network can spontaneously open or close. In addition, the closure rate of an edge also depends on the state of the neighboring edges. This is captured by the cascade function, $f(\mathcal{N}_a, \mathcal{N}_b)$. In this paper, we considered 3 different variations of DBP: 1) SUD-DBP, 2) TRI-DBP, and 3) POD-DBP. The advantage of DBP over other existing network process models is that the equilibrium distribution can be derived in closed-form for arbitrary maximal network.

We proved that the sufficient statistics of SUD-DBP at equilibrium are the number of closed edges, $|E(\mathbf{A})|$, and the number of $P_3$ subgraphs. The sufficient statistics of TRI-DBP are the number of closed edges, $|E(\mathbf{A})|$, and the number of $C_3$ subgraphs. The sufficient statistic of POD-DBP are the number of closed edges, $|E(\mathbf{A})|$, and the number of $C_3$ and $P_4$ subgraphs. As the sufficient statistics are different depending on the cascade function, it is intuitive that the most-probable network of the process may also differ.

In Regime I) Recovery Dominant and Regime IV) Removal Dominant, this is not the case since SUD-DBP, TRI-DBP, and POD-DBP all lead to the consensus most-probable network of $\mathbf{A}_0$ and $\mathbf{A_{\max}}$. This does not apply to Regime II) Cascading Failures and Regime III) Cascading Prevention. We proved that the Most-Probable Configuration Problem can be solved using polynomial-time algorithm in Regime II) because solving the optimization of the most-probable configuration corresponds to finding the minimum of a submodular function. We illustrated using two real-world, heterogenous networks, phase transition behavior as well as the existence of non-degenerate solutions of the Most-Probable Configuration Problem. This means that, for certain processes, some subset of edges is more critical than others. In SUD-DBP, these at-risk edges tend to belong to edges in large star subgraphs; in POD-DBP and TRI-DBP, they tend to belong to edges in triangle subgraphs. 

In Regime III), we proved that, for a subregime of parameters, the solution of the Most-Probable Network Problem corresponds to maximum matching for SUD-DBP. For TRI-DBP, the solution corresponds to triangle-free graphs. For POD-DBP, the solution corresponds to maximum star matching. We showed with DBP and preliminary analysis that the underlying topology and the form of the cascade function have a large impact on the equilibrium behavior of the dynamical process. In the future, we will extend our analysis to transient behaviors as well as considering other cascade functions.

\bibliographystyle{unsrtnat}
\bibliography{refs}


\appendix


\section{Determining $g(E(\mathbf{A}))$ for SUD-DBP}\label{sec:gsud}
For the Sum-Dependent Dynamic Bond Percolation Model, $g(E(\mathbf{A}))$ is the number of paths of length 2 formed by the set of edges, $E(\mathbf{A})$, of the network represented by the adjacency matrix~$\mathbf{A}$. 
The number of walks of length 2 from node $i$ to node $j \neq i$ is
\[
\sum_{i = 1}^N \sum_{j > i} (\mathbf{A}^2)_{i,j}.
\]
Realize that this is also equivalent to the number of paths of length 2 from node $i$ to node $j \neq i$.

\section{Determining $g(E(\mathbf{A}))$ for POD-DBP}\label{sec:gpod}
For the Product-Dependent Dynamic Bond Percolation Model, $g(E(\mathbf{A}))$ is the number of triangles and of paths of length 3 formed by the set of edges, $E(\mathbf{A})$, of the network represented by the adjacency matrix $\mathbf{A}$. 
The number of cycles of length $3$ is
\[
\sum_{i = 1}^N \frac{(\mathbf{A}^3)_{i,i}}{6}.
\]

We need to find the number of paths of length 3. We know that the number of walks of length 3 from node $i$ to node $j \neq i$ is
\[
\sum_{i = 1}^N \sum_{j > i} (\mathbf{A}^3)_{i,j}.
\]
This number is \emph{larger} than the number of paths of length 3 because there are walks from node $i$ to node $j$ that are not paths. Figure~\ref{fig:notpath} illustrates the three cases of walks of length 3 that are not paths of length 3 because the vertices repeat. 

\begin{figure*}[h]
        \centering
        \begin{subfigure}[b]{\textwidth}
                \centering
           \begin{tikzpicture}[->,>=stealth',shorten >=1pt,auto,node distance=3cm,
  		thick,main node/.style={circle,draw,font=\sffamily\small}]

  		\node[main node] (1) {i};
  		\node[main node] (2) [below left of=1] {j};
 		 \node[main node] (3) [below right of=1] {k};

  		\path[every node/.style={font=\sffamily\small}]
    		(1) edge node [left] {$e_3$} (2)
		     edge node [right] {$e_1$} (3)
   		 (3) edge [bend right] node [right] {$e_2$} (1);
		\end{tikzpicture}

                \caption{Walk of Length 3 from Node $i$ to Node $j$: $i, e_1, k, e_2, i, e_3, j$}
                \label{fig:scenario1}
                
        \end{subfigure}%
        ~ 
          
        \begin{subfigure}[b]{\textwidth}
                \centering
      \begin{tikzpicture}[->, >=stealth',shorten >=1pt,auto,node distance= 3 cm,
  		thick,main node/.style={circle,draw,font=\sffamily\small}]

  		\node[main node] (1) {i};
  		\node[main node] (2) [below left of=1] {j};
  		\node[main node] (3) [below right of=1] {k};

  		\path[every node/.style={font=\sffamily\small}]
    		(1) edge node [left] {$e_1$} (2)
    		(2) edge [bend left] node  {$e_2$} (3)
    		(3) edge [bend left] node {$e_3$} (2);
		\end{tikzpicture}
      
                \caption{Walk of Length 3 from Node $i$ to Node $j$: $i, e_1, j, e_2, k, e_3, j$}
                \label{fig:scenario2}
        \end{subfigure}
        ~ 
          
          \begin{subfigure}[b]{\textwidth}
                \centering
	\begin{tikzpicture}[->,>=stealth',shorten >=1pt,auto,node distance=3cm,
  		thick,main node/.style={circle,draw,font=\sffamily\small}]

  		\node[main node] (1) {i};
  		\node[main node] (2) [right of=1] {j};

\path[every node/.style={font=\sffamily\small}]
    		(1) edge node  [below] {$e_1$} (2)
   		 (2) edge [bend right] node {$e_2$} (1)
		 (1) edge [bend left] node  {$e_3$} (2);
		\end{tikzpicture}

                \caption{Walk of Length 3 from Node $i$ to Node $j$: $i, e_1, k, e_2, i, e_3, j$}
                \label{fig:scenario3}
        \end{subfigure}
   
        \caption{Walks of Length 3 that are not Paths of Length 3}\label{fig:notpath}
\end{figure*}
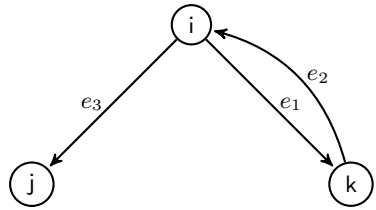
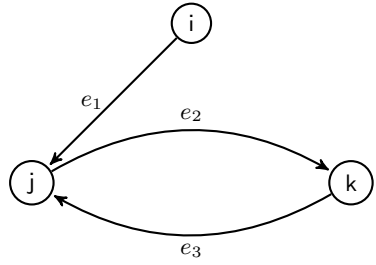
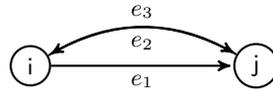

Therefore, the number of paths of length 3 from node $i$ to node $j \neq i$ is
{\scriptsize
\begin{equation}
\begin{split}
&\sum_{i = 1}^N \sum_{j > i} (\mathbf{A}^3)_{i,j} - (\mathbf{A}_{i,j})\sum_{k = 1, k \neq i,j}^N \mathbf{A}_{i,k}  - (\mathbf{A}_{i,j})\sum_{k = 1, k \neq i,j}^N \mathbf{A}_{j,k} -\mathbf{A}_{i,j}\\
= &\sum_{i = 1}^N \sum_{j > i} (\mathbf{A}^3)_{i,j} - (\mathbf{A}_{i,j})\left(\sum_{k = 1, k \neq i,j}^N \mathbf{A}_{i,k} + \sum_{k = 1, k \neq i,j}^N \mathbf{A}_{j,k}  + 1 \right)\\
=  &\sum_{i = 1}^N \sum_{j > i} (\mathbf{A}^3)_{i,j} - (\mathbf{A}_{i,j})\left(  (\mathbf{A}^2)_{i,i} + \sum_{k = 1, k \neq i,j}^N \mathbf{A}_{k,j}  \right).
\end{split}
\end{equation}
}

\noindent This leads to 
\begin{equation*}
\begin{split}
&g(E(\mathbf{A})) = \sum_{i = 1}^N \frac{(\mathbf{A}^3)_{i,i}}{6} + \\
&\sum_{i = 1}^N \sum_{j > i} (\mathbf{A}^3)_{i,j} - (\mathbf{A}_{i,j})\left(  (\mathbf{A}^2)_{i,i} + \sum_{k = 1, k \neq i,j}^N \mathbf{A}_{k,j}  \right).
\end{split}
\end{equation*}


\end{document}